\documentclass[reqno]{amsart}

\makeatletter
\@namedef{subjclassname@2020}{%
  \textup{2020} Mathematics Subject Classification}
 
\makeatother 

\usepackage{amsthm,amssymb,amsfonts,latexsym,mathtools,thmtools,amsmath,lscape}
\usepackage[T1]{fontenc}
\usepackage{tikz-cd} 
\usepackage{enumitem} 
\usepackage{longtable}
\usepackage{array}
\usepackage{multirow}
\usepackage{mathrsfs} 
\usepackage{hyperref} 
\usepackage[all]{xy}
\usepackage{booktabs}
\hypersetup{
    colorlinks=true,
    linkcolor=blue,
    filecolor=blue,      
    urlcolor=blue,
    linktocpage=true
}

\newtheorem{theorem}{Theorem}[section]

\newtheorem{proposition}[theorem]{Proposition}

\theoremstyle{definition}
\newtheorem{definition}[theorem]{Definition}

\newtheorem{remark}[theorem]{Remark}

\theoremstyle{remark}

\numberwithin{equation}{section}

\begin{document}

\title[Smooth Geometry of Diffusion Algebras]{Smooth Geometry of Diffusion Algebras}


\author{Andr\'es Rubiano}
\address{Universidad Nacional de Colombia - Sede Bogot\'a}
\curraddr{Campus Universitario}
\email{arubianos@unal.edu.co}
\address{Universidad ECCI}
\curraddr{Campus Universitario}
\email{arubianos@ecci.edu.co}
\thanks{}


\author{Armando Reyes}
\address{Universidad Nacional de Colombia - Sede Bogot\'a}
\curraddr{Campus Universitario}
\email{mareyesv@unal.edu.co}


\subjclass[2020]{16E45, 16S32, 16S36, 16S37, 16S38, 16W20, 16W50, 58B34}

\keywords{Differentially smooth algebra, integrable calculus, skew polynomial ring, generalized Weyl algebra, diskew polynomial ring, bi-quadratic algebra}

\date{}

\dedicatory{Dedicated to Professor Oswaldo Lezama on the Occasion of His 68th Birthday}

\begin{abstract} 

In this paper, we study the differential smoothness of diffusion algebras. 

\end{abstract}

\maketitle


\section{Introduction}

Isaev et al. \cite{IsaevPyatovRittenberg2001} introduced {\em diffusion algebras} in the context of one-dimensional stochastic processes with exclusion in statistical mechanics. One year later, Pyatov and Twarock \cite{PyatovTwarock2002} presented a construction formalism for these algebras from the mathematical point of view, and proved the results formulated in \cite{IsaevPyatovRittenberg2001}. Just as they said, \textquotedblleft Diffusion algebras play a key role in understanding one-dimensional stochastic processes. In the case of $N$ species of particles with only nearest-neighbor interactions with exclusion on a one-dimensional lattice, diffusion algebras are useful tools in finding expressions for the probability distribution of the stationary state of these processes. Following the idea of matrix product states, the latter are given in terms of monomials built from the generators of a quadratic algebra\textquotedblright\ \cite[p. 3268]{PyatovTwarock2002}. Hinchcliffe in his Ph.D. thesis \cite{Hinchcliffe2005} and different researchers have investigated several ring, theoretical and homological properties of diffusion algebras \cite{Fajardoetal2020, HamidizadehHashemiReyes2020, Levandovskyy2005, ReyesRodriguez2021, ReyesSuarez2016, ReyesSuarez2021, Twarok2003}. 

Our purpose in this paper is to investigate the {\em differential smoothness} defined by Brzezi{\'n}ski and Sitarz \cite{BrzezinskiSitarz2017} of diffusion algebras. Let us say some preliminary words on the subject.

Following Brzezi\'nski and Lomp's ideas in their paper \cite[Section 1]{BrzezinskiLomp2018}, \textquotedblleft the study of smoothness of algebras goes back at least to Grothendieck's EGA \cite{Grothendieck1964}. The concept of a {\em formally smooth commutative} ({\em topological}) {\em algebra} introduced by him was extended to the noncommutative setting by Schelter \cite{Schelter1986}. An algebra is {\em formally smooth} if and only if the kernel of the multiplication map is projective as a bimodule. This notion arose as a replacement of a far too general definition based on the finiteness of the global dimension; Cuntz and Quillen \cite{CuntzQuillen1995} called these algebras {\em quasi-free}. Precisely, the notion of smoothness based on the finiteness of this dimension was refined by Stafford and Zhang \cite{StaffordZhang1994}, where a Noetherian algebra is said to be {\em smooth} provided that it has a finite global dimension equal to the homological dimension of all its simple modules\textquotedblright. In the homological setting, Van den Bergh \cite{VandenBergh1998} called an algebra {\em homologically smooth} if it admits a finite resolution by finitely generated projective bimodules. The characterization of this kind of smoothness for the noncommutative pillow, the quantum teardrops, and quantum homogeneous spaces was made by Brzezi{\'n}ski \cite{Brzezinski2008, Brzezinski2014} and Kr\"ahmer \cite{Krahmer2012}, respectively.

Brzezi{\'n}ski and Sitarz \cite{BrzezinskiSitarz2017} defined other notion of smoothness of algebras, termed {\em differential smoothness} due to the use of differential graded algebras of a specified dimension that admits a noncommutative version of the Hodge star isomorphism, which considers the existence of a top form in a differential calculus over an algebra together with a string version of the Poincar\'e duality realized as an isomorphism between complexes of differential and integral forms. This new notion of smoothness is different and more constructive than the homological smoothness mentioned above. \textquotedblleft The idea behind the {\em differential smoothness} of algebras is rooted in the observation that a classical smooth orientable manifold, in addition to de Rham complex of differential forms, admits also the complex of {\em integral forms} isomorphic to the de Rham complex \cite[Section 4.5]{Manin1997}. The de Rham differential can be understood as a special left connection, while the boundary operator in the complex of integral forms is an example of a {\em right connection}\textquotedblright\ \cite[p. 413]{BrzezinskiSitarz2017}.

Several authors (e.g. \cite{Brzezinski2015, Brzezinski2016, BrzezinskiElKaoutitLomp2010, BrzezinskiLomp2018, BrzezinskiSitarz2017, DuboisVioletteKernerMadore1990, Karacuha2015, KaracuhaLomp2014, ReyesSarmiento2022}) have characterized the differential smoothness of algebras such as the quantum two - and three - spheres, disc, plane, the noncommutative torus, the coordinate algebras of the quantum group $SU_q(2)$, the noncommutative pillow algebra, the quantum cone algebras, the quantum polynomial algebras, Hopf algebra domains of Gelfand-Kirillov dimension two that are not PI, families of Ore extensions, some 3-dimensional skew polynomial algebras, diffusion algebras in three generators, and noncommutative coordinate algebras of deformations of several examples of classical orbifolds such as the pillow orbifold, singular cones and lens spaces. Precisely, in \cite{ReyesSarmiento2022} the second author presented a first approach to the differential smoothness of diffusion algebras on three generators and, due to the relationships between diffusion algebras with 3-dimensional skew polynomial algebras \cite{BellSmith1990, Rosenberg1995}, skew bi-quadratic algebras \cite{Bavula2023} (see also double Ore extensions \cite{ZhangZhang2008, ZhangZhang2009}), and skew PBW extensions \cite{Fajardoetal2020}, and that the smoothness of all these families of algebras has been investigated in \cite{ReyesSarmiento2022, ReyesSuarez2016, RubianoReyes2024DSBiquadraticAlgebras, RubianoReyes2024DSDoubleOreExtensions, RubianoReyes2024DSSPBWKt}, our purpose in this paper is to continue with the research on this topic for diffusion algebras on four and more generators. 

The article is organized as follows. Section \ref{DAPreliminaries} contains the key facts on diffusion algebras in order to set up notation and render this paper self-contained. We adopt the terminology and notation presented by Pyatov and Twarock \cite{PyatovTwarock2002}. Sections \ref{DAThreegenerators} and \ref{DAgeneralgenerators} recall the classification of diffusion algebras on three and $n$ generators, respectively. Next, Section \ref{DefinitionsandpreliminariesDSA} contains definitions and preliminaries on differential smoothness of algebras; we also set up notation necessary for the rest of the paper. In Section \ref{DSDANgenerators} we present the original results of the paper. Tables \ref{Diffusion4} and \ref{Diffusion4(2)} present the diffusion algebras on four generators while Tables \ref{Diffusion5(1)}, \ref{Diffusion5(2)}, \ref{Diffusion5(3)}, \ref{Diffusion5(4)}, \ref{Diffusion5(5)}, \ref{Diffusion5(6)}, \ref{Diffusion5(7)}, \ref{Diffusion5(8)}, \ref{Diffusion5(9)} and \ref{Diffusion5(10)} contain diffusion algebras on five generators. Our key results are Theorems \ref{DSdiffalgN} and \ref{noDSDiffN} since these  describe explicitly those diffusion algebras that are differentially smooth, and also those families of algebras which are not, respectively.

Throughout the paper, $\mathbb{N}$ denotes the set of natural numbers including zero. The word ring means an associative ring with identity not necessarily commutative. All vector spaces and algebras (always associative and with unit) are over a fixed field $\Bbbk$. As usual, the symbols $\mathbb{R}$ and $\mathbb{C}$ denote the fields of real and complex numbers, respectively. 

\section{Diffusion algebras}\label{DAPreliminaries}

Pyatov and Twarock \cite{PyatovTwarock2002} studied diffusion algebras from the mathematician's point of view and proved a construction theorem for diffusion algebras. Let us recall the details of their treatment.

Let $\alpha, \beta$ be two elements belonging to the set $I_N := \{1, \dotsc, n\}$ with $\alpha < \beta$. Consider quadratic relations of the form
\begin{equation}\label{PyatovTwarock2002(1)}
    g_{\alpha \beta} D_{\alpha} D_{\beta} - g_{\beta \alpha} D_{\beta} D_{\alpha} = x_{\beta} D_{\alpha} - x_{\alpha} D_{\beta},
\end{equation}

with $g_{\alpha \beta} \in \mathbb{R} \ \backslash \ \{0\}, \ g_{\beta \alpha} \in \mathbb{R}$, and $x_{\alpha}, x_{\beta} \in \mathbb{C}$.

\begin{definition}[{\cite[Definition 1.1]{PyatovTwarock2002}}]
An algebra with set of generators given by $\left\{ D_{\alpha} \mid \alpha \in I_N\right\}$ and relations of type {\rm (}\ref{PyatovTwarock2002(1)}{\rm )} is called {\em diffusion algebra}, if it admits a linear PBW-basis of ordered monomials of the form
\begin{equation}\label{PyatovTwarock2002(2)}
    D_{\alpha_1}^{k_1} D_{\alpha_2}^{k_2} \dotsb D_{\alpha_N}^{k_N},\quad {\rm with} \quad k_j \in \mathbb{N} \quad {\rm and} \quad \alpha_1 > \alpha_2 > \dotsb > \alpha_N.
\end{equation}
\end{definition}

Due to physical reasons only relations with positive coefficients $g_{\alpha \beta} \in \mathbb{R}_{>0}$ and $g_{\beta \alpha} \in \mathbb{R}_{\ge 0}$ ($\alpha < \beta$) are relevant because they are interpreted as hopping rates in stochastic models \cite[p. 3268]{PyatovTwarock2002}.

Note that the requirement of having a PBW basis (\ref{PyatovTwarock2002(2)}) implies conditions on the coefficients $g_{\alpha \beta}$ and $x_{\alpha}$ in (\ref{PyatovTwarock2002(1)}) according the the {\em Diamond Lemma} in ring theory formulated by Bergman \cite{Bergman1978}. This means that we have a criterion to verify under which conditions the relations in (\ref{PyatovTwarock2002(1)}) are of PBW type: this is the case precisely if each subset of three generators $\left \{D_{\alpha}, D_{\beta}, D_{\gamma}\right\}$ with ordering $\alpha < \beta < \gamma$ is reduction unique with respect to the ordering, that is if the two ways of reducing the monomial $D_{\alpha} D_{\beta} D_{\gamma}$ to the monomial $D_{\gamma} D_{\beta} D_{\alpha}$ lead to the same result when expressed in the PBW basis (\ref{PyatovTwarock2002(2)}).

Just as Pyatov and Twarock \cite[p. 3269]{PyatovTwarock2002} asserted, the task of deriving all diffusion algebras with $N$ generators reduces to the following two steps:
\begin{enumerate}
    \item [\rm (1)] Find all diffusion algebras with three generators.
    \item [\rm (2)] Find all algebras with $N$ generators such that each subset of three generators coincides with one of the cases listed before.
\end{enumerate}

As it can be seen, the step (1) is equivalent to find those coefficients $g_{\alpha \beta}$ and $x_{\alpha}$ in (\ref{PyatovTwarock2002(1)}) for which a set $\{D_{\gamma}, D_{\beta}, D_{\alpha}\}$ of three generators is reduction unique in the above sense. The list of diffusion algebras of three generators is given in Section \ref{DAThreegenerators}. The second step is a combinatorial problem: it requires one to combine in a consistent way the three generators algebras listed before to algebras with $N$ generators for general $N > 3$.

Pyatov and Twarock considered a constructive method to approach the second step, the so-called {\em blending procedure} (Section \ref{DAgeneralgenerators}) which is an inductive procedure for the construction of diffusion algebras: \textquotedblleft It uses the three generators cases and augments them to larger units by attaching further generators in accordance with the requirements of the diamond lemma, then giving a prescription of how theser larger building blocks may be glued together in order to obtain a general diffusion algebra of $N$ generators\textquotedblright\ \cite[p. 3269]{PyatovTwarock2002}.

\subsection{Diffusion algebras on three generators}\label{DAThreegenerators}

Consider a set $\left\{ D_{\alpha}, D_{\beta}, D_{\gamma} \right\}$ of three generators with an ordering induced by the ordering of the index set $\alpha < \beta < \gamma$ and relations as in (\ref{PyatovTwarock2002(1)}). Since $g_{\alpha \beta} \neq 0$ for all $\alpha, \beta \in I_3 = \left\{ \alpha, \beta, \gamma, \right\}$ with $\alpha < \beta$ by assumption, we get the relations of the algebra given by
\begin{align}
    D_{\alpha} D_{\beta} = &\ q_{\beta \alpha} D_{\beta} D_{\alpha} + x_{\beta}^{\alpha \beta} D_{\alpha} - x_{\alpha}^{\alpha \beta} D_{\beta}, \notag \\
    D_{\alpha} D_{\gamma} = &\ q_{\gamma \alpha} D_{\gamma} D_{\alpha} + x_{\gamma}^{\alpha \gamma} D_{\alpha} - x_{\alpha}^{\alpha \gamma} D_{\gamma}, \quad {\rm and} \label{PyatovTwarock2002(3)} \\
    D_{\beta} D_{\gamma} = &\ q_{\gamma \beta} D_{\gamma} D_{\beta} + x_{\gamma}^{\beta \gamma} D_{\beta} - x_{\beta}^{\beta \gamma} D_{\gamma}, \notag \\
\end{align}

where $q_{ij} := \frac{g_{ji}}{g_{ij}}, \ x_{\gamma}^{ij} := \frac{x_{\gamma}}{g_{ij}}$ for $i, j, k \in \left\{\alpha, \beta, \gamma \right\}$ with $i < j$. Using (\ref{PyatovTwarock2002(3)}), any monomial can be expressed in terms of the PBW basis (\ref{PyatovTwarock2002(2)}), and this is well defined if we apply (\ref{PyatovTwarock2002(3)}) in different orders and obtain the same result, that is, if the reductions
\begin{equation}\label{PyatovTwarock2002(4)} 
    D_{\alpha} D_{\beta} D_{\gamma} \xrightarrow[]{} D_{\beta} D_{\alpha} D_{\gamma} \xrightarrow[]{}  D_{\beta} D_{\gamma} D_{\alpha} \xrightarrow[]{} D_{\gamma} D_{\beta} D_{\alpha}
\end{equation}

and

\begin{equation}\label{PyatovTwarock2002(5)} 
    D_{\alpha} D_{\beta} D_{\gamma} \xrightarrow[]{} D_{\alpha} D_{\gamma} D_{\beta} \xrightarrow[]{}  D_{\gamma} D_{\alpha} D_{\beta} \xrightarrow[]{} D_{\gamma} D_{\beta} D_{\alpha}
\end{equation}

using (\ref{PyatovTwarock2002(3)}) coincide when expressed in the PBW basis. These equalities lead to restrictions on the coefficients $g_{\alpha \beta}$ and $x_{\alpha}$ in (\ref{PyatovTwarock2002(1)}) given by a set of six equations and their solutions determine all diffusion algebras of three generators. 

Next, we recall the list of six equations. We assume $\alpha < \beta < \gamma$ and $x_{\beta} \neq 0$ for $j \in \left\{ \alpha, \beta, \gamma \right\}$.

\begin{enumerate}
\item {\em The case of} $A_I$: 
\begin{align*}
  gD_{\alpha}D_{\beta}  - gD_{\beta} D_{\alpha} = &\ x_{\beta}D_{\alpha}-x_{\alpha}D_{\beta},\\
  gD_{\alpha}D_{\gamma} - g D_{\gamma}D_{\alpha} = &\ x_{\gamma}D_{\alpha} - x_{\alpha}D_{\gamma}, \quad {\rm and}\\
 gD_{\beta}D_{\gamma} - gD_{\gamma} D_{\beta} = &\ x_{\gamma}D_{\beta} - x_{\beta}D_{\gamma}, 
\end{align*}

where $g \neq 0$.

\item {\em The case of} $A_{II}$:
\begin{align*}
  g_{\alpha \beta} D_{\alpha}D_{\beta} = &\ x_{\beta} D_{\alpha} - x_{\alpha} D_{\beta},\\
  g_{\alpha \gamma} D_{\alpha} D_{\gamma} = &\ x_{\gamma} D_{\alpha} - x_{\alpha} D_{\gamma}, \quad {\rm and}\\
 g_{\beta \gamma}D_{\beta} D_{\gamma} = &\ x_{\gamma} D_{\beta} - x_{\beta} D_{\gamma},
\end{align*}

where $g_{ij} := g_i - g_j$ with $g_i \neq g_j$ for all $i, j\in \{\alpha, \beta, \gamma\}$ with $i < j$.

\item {\em The case of} $B^{(1)}$:
\begin{align*}
    g_{\beta} D_{\alpha} D_{\beta} - (g_{\beta}  -\Lambda)D_{\beta} D_{\alpha}  = &\ -x_{\alpha} D_{\beta} ,\\
    gD_{\alpha} D_{\gamma}  - (g - \Lambda)D_{\gamma} D_{\alpha}  = &\ x_{\gamma} D_{\alpha}  - x_{\alpha} D_{\gamma}, \quad {\rm and}\\
    g_{\beta} D_{\beta} D_{\gamma}  - (g_{\beta}  - \Lambda)D_{\gamma} D_{\beta}  = &\ x_{\gamma} D_{\beta} ,
\end{align*}

where $g, g_{\beta}  \neq 0$. 

 \item {\em The case of} $B^{(2)}$: 
 \begin{align*}
   g_{{\alpha} {\beta} }D_{\alpha} D_{\beta}  = &\ -x_{\alpha} D_{\beta} ,\\
    g_{{\alpha} {\gamma} }D_{\alpha} D_{\gamma}  - g_{{\gamma} {\alpha} }D_{\gamma} D_{\alpha}  = &\ x_{\gamma} D_{\alpha} - x_{\alpha} D_{\gamma}, \quad {\rm and} \\
    g_{{\beta} {\gamma} }D_{\beta} D_{\gamma}  = &\ x_{\gamma} D_{\beta} ,
\end{align*}

where $g_{{\alpha} {\beta} }, g_{{\alpha} {\gamma} }, g_{{\beta} {\gamma} } \neq 0$.

\item {\em The case of} $B^{(3)}$:
 \begin{align*}
    gD_{\alpha} D_{\beta}  - (g-\Lambda)D_{\beta} D_{\alpha}  = &\ x_{\beta} D_{\alpha}  -x_{\alpha} D_{\beta} ,\\
    g_{\gamma} D_{\alpha} D_{\gamma}  = &\ -x_{\alpha} D_{\gamma}, \quad {\rm and}\\
    (g_{\gamma}  - \Lambda)D_{\beta} D_{\gamma}   = &\ - x_{\beta} D_{\gamma} ,
\end{align*}

where $g\neq 0$ and $g_{\gamma} \neq 0, \Lambda$.

\item {\em The case of} $B^{(4)}$:
\begin{align*}
    (g_{\alpha}  - \Lambda)D_{\alpha} D_{\beta}  = &\ x_{\beta} D_{\alpha} ,\\
    g_{\alpha} D_{\alpha} D_{\gamma}  = &\ x_{\gamma} D_{\alpha}, \quad {\rm and}\\
    gD_{\beta} D_{\gamma}  - (g-\Lambda)D_{\gamma} D_{\beta}  = &\ x_{\gamma} D_{\beta}  - x_{\beta} D_{\gamma} ,
\end{align*}
  where $g\neq 0$ and $g_{\alpha} \neq 0, \Lambda$.

\item {\em The case of} $C^{(1)}$: 
 \begin{align*}
    g_{\beta} D_{\alpha} D_{\beta}  - (g_{\beta}  -\Lambda)D_{\beta} D_{\alpha}  = &\ -x_{\alpha} D_{\beta} ,\\
    g_{\gamma} D_{\alpha} D_{\gamma}  - (g_{\gamma}  - \Lambda)D_{\gamma} D_{\alpha}  = &\ - x_{\alpha} D_{\gamma}, \quad {\rm and}\\
    g_{{\beta} {\gamma} }D_{\beta} D_{\gamma}  - g_{{\gamma} {\beta} }D_{\gamma} D_{\beta}  = &\ 0,
\end{align*}
where $g_{\beta} , g_{\gamma} , g_{{\beta} ,{\gamma} }\neq 0$.

\item {\em The case of} $C^{(2)}$: 
 \begin{align*}
    g_{{\alpha} {\beta} }D_{\alpha} D_{\beta}  - g_{{\beta} {\alpha} }D_{\beta} D_{\alpha}  = &\ - x_{\alpha} D_{\beta} ,\\
    g_{{\alpha} {\gamma} }D_{\alpha} D_{\gamma}  - g_{{\gamma} {\alpha} }D_{\gamma} D_{\alpha}  = &\ - x_{\alpha} D_{\gamma}, \quad {\rm and}\\
    D_{\beta} D_{\gamma}   = &\ 0,
\end{align*}

where $g_{{\alpha} {\beta} }, g_{{\alpha} {\gamma} } \neq 0$.

\item {\em The case of} $D$: With $q_{ji}: = \frac{g_{ji}}{g_{ij}}$, where $i, j\in \{\alpha, \beta, \gamma\}$ (recall that $g_{ij} \neq 0$, for $i < j$), we have that
 \begin{align*}
    D_{\alpha} D_{\beta}  - q_{{\beta} {\alpha} }D_{\beta} D_{\alpha}  = &\ 0, \\
    D_{\alpha} D_{\gamma}  - q_{{\gamma} {\alpha} }D_{\gamma} D_{\alpha}  = &\ 0, \quad {\rm and}\\
    D_{\beta} D_{\gamma}  - q_{{\gamma} {\beta} }D_{\gamma} D_{\beta}  = &\ 0.
\end{align*} 
\end{enumerate}

The division into algebras $A, B, C$ and $D$ reflects the number of coefficients $x_j, j \in \left\{ \alpha, \beta, \gamma \right\}$ being zero in the expression (\ref{PyatovTwarock2002(1)}). As it was shown above, for these algebras, none, one, two, or all three of the coefficients $x_i$ vanish, respectively.

\subsection{Diffusion algebras on \texorpdfstring{$N$}{Lg} generators}\label{DAgeneralgenerators}

Consider the following decomposition of the index set $I_N = \left\{ 1, \dotsc, N \right\}$: 
\begin{equation}\label{PyatovTwarock2002(15)(16)}
    I_N = I \cup R, \quad {\rm with} \quad I := \left\{\alpha \in I_N \mid x_{\alpha} \neq 0\right\} \ {\rm and} \ R:= \left\{ \alpha \in I_N \mid x_{\alpha} = 0 \right\}.
\end{equation}

\begin{definition}\label{PyatovTwarock2002Definitions3.1and3.2}
\begin{enumerate}
    \item \cite[Definition 3.1]{PyatovTwarock2002} Normal orderings of two generators $D_{\alpha}$ and $D_{\beta}$ is defined as
\begin{equation}\label{PyatovTwarock2002(17)}
    \left(D_{\alpha} D_{\beta}\right) := \begin{cases}
D_{\alpha} D_{\beta}, & {\rm if} \ \alpha < \beta \\ D_{\beta} D_{\alpha}, & {\rm if} \ \beta < \alpha.
    \end{cases}
\end{equation}

\item \cite[Definition 3.2]{PyatovTwarock2002} For $\alpha < \beta$, consider the following notation:
\begin{equation}\label{PyatovTwarock2002(18)}
    [D_{\alpha}, D_{\beta}]_{q_{\beta \alpha}} := D_{\alpha} D_{\beta} - q_{\beta \alpha} D_{\beta} D_{\alpha},
\end{equation}

where the index at the commutator is referring to the coefficients $q_{\beta \alpha}$ in terms of which the commutator is defined.
\end{enumerate}
\end{definition}

Considering notation in Definition \ref{PyatovTwarock2002Definitions3.1and3.2}, the set $R$ is subdivided into nonintersecting and nonempty subsets
\begin{equation}\label{PyatovTwarock2002(19)}
    R := R_1 \cup R_2 \cup \dotsb \cup R_{M_R}
\end{equation}

according to the following requirements:
\begin{itemize}
    \item Relations between generators from the sets $R_{a}$ and $R_b$ for $a\neq b$ are given by
\begin{equation}\label{PyatovTwarock2002(20)}
     :D_{r_1} D_{r_2}: = 0, \quad {\rm for\ all} \ r_1 \in R_a \ {\rm and} \ r_2 \in R_b.   
    \end{equation}

    \item Relations within a set $R_a$ such that $|R_a| \ge 2$ are given by
\begin{equation}\label{PyatovTwarock2002(21)}
        [D_{r_1}, D_{r_2}]_{q_{r_2 r_1}} = 0, \quad {\rm for\ all} \ r_1, r_2 \in R_a \ {\rm with} \ r_1 < r_2,
    \end{equation}

    where the coefficients in (\ref{PyatovTwarock2002(21)}) are subject to the condition opposite to (\ref{PyatovTwarock2002(20)}), that is: for any subdivision $R_a = R' \cup R''$ into two nonintersecting and nonempty parts $R'$ and $R''$, 
    \begin{equation}\label{PyatovTwarock2002(22)}
        {\rm there\ exists}\ r_1 \in R' \ {\rm and} \ r_2\in R''\ {\rm such\ that} \ g_{r_1 r_2} g_{r_2 r_1} \neq 0.
    \end{equation}

    This means that for any pair of indices $r, s \in R_a$ there exists a finite sequence $\left\{ r_k \in R_a \mid k = 1, \dotsc, n \right\}$ such that $r_1 = r, r_N = s$ and
\begin{equation}\label{PyatovTwarock2002(23)}
        \prod_{k = 1}^{N-1} g_{r_kr_{k+1}} g_{r_{k+1} r_k} \neq 0.
    \end{equation}

    Relations (\ref{PyatovTwarock2002(22)}) and (\ref{PyatovTwarock2002(23)}) may be represented graphically via a {\em connectivity condition} on an ordered graph the vertices of which are labeled by the indices $r\in R_a$ and the edges connect only those vertices $r_1 < r_2$ for which the condition $q_{r_2 r_1} \neq 0$ is satisfied.
\end{itemize}

    Furthermore, for $|I| \ge 2$ the set $R$ is split into two sets $S$ and $T$ as follows:

    For any $R_a \subset R$, let
\begin{equation}\label{PyatovTwarock2002(24)}
        R_a := \begin{cases}
S_a, & {\rm if} \ {\rm there\ exist} \ r\in R_a \ {\rm and} \ i\in I \ {\rm with} \ g_{ir} g_{ri} \neq 0, \\
T_a & {\rm otherwise}.
        \end{cases}
    \end{equation}

Suppose that the $M_R$ sets $R_a$ in (\ref{PyatovTwarock2002(19)}) split into $M_S$ sets $S_a$ and $M_T$ sets $T_a$. Then $M_R = M_S + M_T$. Number these sets as $S_a$ for $a = 1, \dotsc, M_S$, and $T_a$ for $a = 1, \dotsc, M_T$, and introduce
\begin{equation}\label{PyatovTwarock2002(25)}
S:= \bigcup_{a = 1}^{M_S} S_a \quad {\rm and} \quad T:= \bigcup_{a = 1}^{M_T} T_a.
\end{equation}

The decomposition of the set $S$ into subsets $S_a$ has been used in the definition of the set $S$, and this will not be consider. However, the structure of the set $T$ is crucial and needs further refinement.

For every $T_a \subset T$, let
\begin{equation}\label{PyatovTwarock2002(26)}
    T_a:= \begin{cases}
T_a^{\bullet}, \quad {\rm if} \ \exists \ i, j\in I \ {\rm with} \ i < j \ {\rm such\ that} \ T_a \subset \{ i + 1, i  + 2, \dotsc, j - 1\} \ {\rm and} \\ \quad \quad \quad \quad \quad \quad \quad \quad \quad \quad \quad \quad \quad \quad \quad \quad \quad I \cap \{ i + 1, i + 2, \dotsc, j - 1\} = \emptyset, \\
T_a^{\circ}, \quad {\rm otherwise}.
\end{cases}
\end{equation}

In this way, we write
\begin{equation}\label{Taeq}
T = \left\{ T_a^{\bullet} \mid a = 1, \dotsc, M_T^{\bullet} \right\} \bigcup \left\{ T_a^{\circ} \mid a = 1, \dotsc, M_T^{\circ} \right\} \quad {\rm with} \ M_T = M_T^{\bullet} + M_T^{\circ}.
\end{equation}

\subsubsection{List of diffusion algebras on \texorpdfstring{$N$}{Lg} generators}

From now on, the expression \textquotedblleft generators of a set $I, S, T$, or $R$\textquotedblright\ means the generators indexed by elements from the corresponding set.

\begin{definition}[{\cite[Definition 3.3]{PyatovTwarock2002}}]
A set of three generators $\left\{ D_x, D_y, D_z\right\}$ with $x \in X, y\in Y$ and $z\in Z$, where $X, Y$ and $Z$ are any of the sets $I, R, S$ and $T$ or any set in their decomposition will be called a {\em triplet} ({\em of type}) $\{X, Y, Z\}$. 
\end{definition}

As it can be seen, any triplet of type $\{I, I, I\}$ in a diffusion algebra of $N\ge 3$ generators gives rise to a diffusion algebra of type $A_I$ or $A_{II}$, any triple of type $\{I, I, R\}$ to a diffusion algebra of type $B^{(1)}, B^{(2)}, B^{(3)}$ or $B^{(4)}$, any triplet of type $\{I, R, R\}$ to a diffusion algebra of type $C^{(1)}$ or $C^{(2)}$, and any triplet of type $\{R, R, R\}$ to a diffusion algebra of type $D$.

\begin{proposition}[{\cite[Lemma 4.3]{PyatovTwarock2002}}]
For any diffusion algebra {\rm (1)} with $N \ge 3$ generators, the following statements hold:
\begin{enumerate}
    \item [\rm (1)] If $|I|\geq 3$, then all subalgebras corresponding to triplets of type $\{I, I, I\}$ are of the same type, which is either $A_I$ {\rm (}that is, $g_{ij} = g$ for all $i, j \in I${\rm )} or $A_{II}$ (that is, $g_{ji} = 0, \ g_{ij} = g_i - g_j, \ g_i \neq g_j$ for all $i < j$).

    \item [\rm (2)] If $|I| \ge 3$ and all subalgebras corresponding to triplets $\{I, I, I\}$ are of type $A_I$, then for any $s\in S$ and for all $i \in I$, we have that
    \begin{equation}\label{PyatovTwarock2002(27)}
         g_{is} = g_{si} = g_s.
    \end{equation}
   
    \item [\rm (3)] If $|I| \ge 3$ and all subalgebras corresponding to triplets $\{ I, I, I\}$ are of type $A_{II}$, then $S = \emptyset$.

    \item [\rm (4)] Let $|I| \ge 2$. For any $i\in I$ and every $t, t' \in T_a$ ($T_a$ means both $T_a^{\circ}$ and $T_a^{\bullet}$) with $t < i$ and $t' > i$, the coefficients $g_{ti}$ and $g_{it'}$ depend only on the index $a$ of the set $T_a$ and not on the individual indices $t$ or $t'$. If $t, t' \in T_a^{\circ}$ one furthermore has $g_{ti} = - g_{it'}$.

    For any $i < j$ and every $t, t' \in T_a$: $t < i$ and $t' > j$,
    \begin{equation}
        g_{ti} + \Lambda_{ij} = g_{tj} \quad {\rm and} \quad g_{it'} = g_{jt'} + \Lambda_{ij}, \quad {\rm where}\ \Lambda_{ij} := g_{ij} - g_{ji}.
    \end{equation}

    \item [\rm (5)] Let $|I| = 1$. Denote the only index in $I$ as ${\bf i}$ in order to stress that it is not a running index. For all $r\ in R_a$ one has
    \begin{equation}\label{I1}
        g_{{\bf i}r}-g_{r{\bf i}}=\Lambda_a.
    \end{equation}
    Note that both the left- and the right-hand sides of Relation {\em (}\ref{I1}{\em )} depend only on the index $a$ of the set $R_a$ and not on the individual index $r$.
\end{enumerate}
\end{proposition}

Diffusion algebras with $N$ generators are listed as five families of algebras: $A_I, A_{II}, B, C$ and $D$. As in the case of $N = 3$ the number of nonzero coefficients $x_{\alpha}$, or equivalently, the cardinality of the set $I$ is used as a criterion for separating diffusion algebras into families of the types $A (|I| \ge 3)$, $B (|I| \ge 2)$, $C (|I| = 1)$ or $D (|I| = 0)$. Type $A$ algebras are separated further into two families $A_I$ and $A_{II}$ depending on the number of nonzero coefficients $g_{ij}$ with indices $i, j$ in the set $I$.

Next, we will see that different algebras in the families are obtained in dependence on the choice of the decomposition of the set $I_N = \{ 1, 2, \dotsc, N\}$ into {\em ordered subsets} $I, S, T_a^{\rm \circ}\ (a = 1, \dotsc, M_T^{\circ}), T_b^{\bullet} \ (b = 1, \dotsc, M_T^{\bullet})$ (or $R_a$, $a = 1, \dotsc, M_R$ for $N_I = 1$) as well as on the choice of coefficients in their defining relations. Next, we consider a notation for diffusion algebras where the corresponding decomposition of the set $I_N$ is given explicitly as argument to the family symbol. The subscript indices $a$ and $b$ in our notation are treated as running ones so that, e.g.
\[
A_I(I, S, T_a^{\circ}, T_b^{\bullet}) \equiv A_I (I, S, T_a^{\circ}, \dotsc, T_{M_T^{\circ}}^{\circ}, T_b^{\bullet}, \dotsc, T_{M_T^{\bullet}}^{\bullet}),
\]

where $I_N = I \cup S \cup \left( \bigcup_{a = 1}^{M_T^{\circ}} T_a^{\circ} \right) \cup \left( \bigcup_{b = 1}^{M_T^{\bullet}} T_a^{\bullet} \right)$, and $I, S, T_a^{\circ}$ and $T_b^{\bullet}$ are mutually nonintersecting ordered subsets in $I_N$. Notice that the values of the coefficients $g_{\alpha \beta}$ are not shown explicitly in these notations, so that this notation displays connective components in a variety of diffusion algebras rather than the particular algebras.

Next, relations in (\ref{PyatovTwarock2002(31)}) - (\ref{PyatovTwarock2002(35)}) below are to be complemented by relations (\ref{PyatovTwarock2002(20)}), (\ref{PyatovTwarock2002(21)}) for the elements of the subset $R$ together with the conditions (\ref{PyatovTwarock2002(22)}) or (\ref{PyatovTwarock2002(23)}) on the coefficients involved.

\begin{proposition}[{\cite[Theorem 3.5]{PyatovTwarock2002}}]\label{PyatovTwarock2002Theorem3.5}
    The following list contains all possible diffusion algebras with $N$ generators:
\begin{enumerate}
    \item [\rm (1)] Diffusion algebras of type $A_I (I, S, T_a^{\circ}, T_b^{\bullet})$ with $|I| \ge 3$:
    \begin{align}
        gD_iD_j - gD_jD_i = &\ x_jD_i - x_iD_j, \quad\ {\rm for\ all}\ i, j \in I, \notag \\
    g_sD_s D_i - g_sD_iD_s = &\ x_i D_s, \quad {\rm for\ all} \ s\in S, i\in I, \notag \\
    g_a^{\circ} : D_i D_t := &\ -x_iD_t, \quad {\rm for \ all} \ a, t \in T_a^{\circ}, \ i \in I, \label{PyatovTwarock2002(31)} \\
    g_b^{+} D_i D_t = &\ -x_i D_t, \quad {\rm for\ all} \ b, t \in T_b^{\bullet}, \ {\rm and\ every}\ i < t, \notag \\
    g_b^{-} D_t D_i = &\ x_i D_t, \quad {\rm for\ all} \ b, t \in T_b^{\bullet},\ {\rm and\ every}\ i > t, \notag
    \end{align}

    where $g, g_s, g_a^{\circ}, g_{b}^{\pm} \neq 0$.

    \item [\rm (2)] Diffusion algebras of type $A_{II} (I, T_a^{\circ}, T_b^{\bullet}), \ |I| \ge 3$:
    \begin{align}
        (g_i - g_j) D_i D_j = &\ x_j D_i - x_iD_j, \quad {\rm for\ all}\ i < j, \notag \\
        (g_i + g_a^{\circ}): D_i D_t := &\ -x_i D_t, \quad {\rm for\ all}\ a, t \in T_a^{\circ}, \ i < t,  \notag \\
         (g_i + g_a^{+})  D_iD_t = &\ -x_iD_t, \quad {\rm for\ all}\ b, t \in T_b^{\bullet}, \ i < t, \label{PyatovTwarock2002(32)} \\
          (g_b^{-} - g_i): D_t D_i : = &\ x_i D_t, \quad {\rm for\ all}\ b, t \in T_b^{\bullet}, \ i > t, \notag 
    \end{align}

where $g_i \neq g_j$ for $i \neq j$ and $g_i \notin \{g_a^{\circ}, \pm g_b^{\pm}\}$.

\item [\rm (3)] Diffusion algebras of type $B (I = \{{\bf i}, {\bf j}\}, S, T_a^{\circ}, T_b^{\bullet})$:
We use the notation ${\bf i}$ and ${\bf j}$ with ${\bf i} < {\bf j}$ for the two elements of the set $I$ to emphasize that they are not running indices. Note that ${\bf i} < t < {\bf j}$ for all $t \in T_b^{\bullet}$ in this case.
\begin{align}
    gD_{\bf i} D_{\bf j} - (g - \Lambda) D_{\bf j} D_{\bf i} = &\ x_{\bf j} D_{\bf i} - x_{\bf i} D_{\bf j}, \notag \\
    g_s D_{\bf i} D_s - (g_s - \Lambda) D_s D_{\bf i} = &\ -x_{\bf i} D_s, \quad {\rm for\ all}\ s\in S, \notag \\
    g_s D_s D_{\bf j} - (g_s - \Lambda)D_{\bf j} D_s = &\ x_{\bf j} D_s, \quad {\rm for\ all} \ s\in S, \notag \\
    g_a^{\circ}: D_{\bf i} D_t := &\ - x_{\bf i} D_t, \quad {\rm for\ all} \ t \in T_a^{\circ}, \label{PyatovTwarock2002(33)} \\
    (g_a^{\circ} - \Lambda) : D_{\bf j} D_t := &\ -x_{\bf j} D_t, \quad {\rm for\ all}\ t\in T_a^{\circ}, \notag \\
    g_b^{+} D_{\bf i} D_t = &\ -x_{\bf i} D_t, \quad {\rm for\ all}\ t \in T_b^{\bullet}, \notag \\
    g_b^{-}D_t D_{\bf j} = &\ x_{\bf j} D_t, \quad {\rm for\ all} \ t \in T_b^{\bullet}, \notag
\end{align}

where $g\neq 0, \ g_s \neq 0$ for all $s$ and $g_s \neq \Lambda$ for $s$ such that either $s < {\bf i}$ or $s > {\bf j}$, $g_a^{\circ} \notin \{0, \Lambda\}$ and $g_b^{\pm} \neq 0$.

\item [\rm (4)] Diffusion algebras of type $C (I = \{{\bf i}\}, R_a)$:
The only element of $I$ is denoted by ${\bf i}$, and hence
\begin{equation}\label{PyatovTwarock2002(34)}
    g_r D_{\bf i} D_r - (g_r - \Lambda_a) D_r D_{\bf i} = -x_{\bf i}D_r, \quad {\rm for\ all} \ r\in R_a,
\end{equation}

where $g_r \neq 0$ for $r < {\bf i}$, and $g_r \neq \Lambda_a$ for $r > {\bf i}$.

\item [\rm (5)] Diffusion algebras of type $D (R)$:
\begin{equation}\label{PyatovTwarock2002(35)}
D_r D_s - q_{sr} D_s D_r = 0, \quad {\rm for \ all} \ r, s \in R \ {\rm with}\ r < s.
\end{equation}
\end{enumerate}
\end{proposition}

\begin{remark}\label{remarkdifalg}
\begin{enumerate}
    \item [\rm (i)] Hinchcliffe in his PhD thesis \cite[Definition 2.1.1]{Hinchcliffe2005} considered the following notation for diffusion algebras. Let $R$ be the algebra generated by $n$ indeterminates $x_1, x_2$, $\dotsc$, $x_n$ over $\mathbb{C}$ subject to relations 
    \[
    a_{ij}x_ix_j - b_{ij}x_jx_i = r_jx_i - r_ix_j
    \]
    
    whenever $i < j$, for some parameters $a_{ij}\in \mathbb{C}\ \backslash\ \{0\}$, for all $i < j$ and $b_{ij}, r_i \in \mathbb{C}$, for all $i < j$. He defined the {\em standard monomials} to be those of the form $x_n^{I_N}x_{n-1}^{i_{n-1}}\dotsb x_2^{i_2}x_1^{i_1}$. $R$ is called a {\em diffusion algebra} if it admits a {\em PBW} {\em basis of these standard monomials}. In other words, $R$ is a diffusion algebra if these standard monomials are a $\mathbb{C}$-vector space basis for $R$. If all the elements $q_{ij}:= \frac{b_{ij}}{a_{ij}}$'s are non-zero, then the diffusion algebras have a PBW basis in any order of the indeterminates \cite[Remark 2.1.6]{Hinchcliffe2005}.
    
From his definition, a diffusion algebra generated by $n$ indeterminates has Gelfand-Kirillov dimension $n$ since because of the PBW basis, the vector subspace consisting of elements of total degree at most $l$ is isomorphic to that of a commutative polynomial ring in $n$ indeterminates. Notice that a diffusion algebra in one indeterminate is precisely a commutative polynomial ring in one indeterminate. A diffusion algebra with $x_t = 0$, for all $t = 1,\dotsc, n$, is a {\em multiparameter quantum affine} $n-${\em space}.

\item [\rm (ii)] Fajardo et al. \cite[Section 2.4]{Fajardoetal2020} studied ring-theoretical properties of a graded version of these algebras. The \textit{diffusion algebras type 2} are affine algebras $\mathcal{D}$ generated by $2n$ variables $\{D_1,\dotsc ,D_n,x_1,\dotsc ,x_n\}$ over $\Bbbk$ that admit a linear PBW basis of ordered monomials of the form $B_{\alpha_{1}}^{k_{1}}B_{\alpha_{2}}^{k_{2}}\cdots B_{\alpha_{n}}^{k_{n}}$ with $B_{\alpha_{i}}\in \{D_1,\dotsc ,D_n,x_1,\dotsc ,x_n\}$, for all $i\leq 2n$, $k_{j}\in\mathbb{N}$, and $\alpha_1>\alpha_2>\cdots > \alpha_n$, such that for all $1\leq i<j\leq n$, there exist elements $\lambda_{ij}\in \Bbbk^*$ satisfying the relations
\begin{equation}\label{equationsrule}
    \lambda_{ij}D_iD_j-\lambda_{ji}D_jD_i=x_jD_i-x_iD_j.
\end{equation}

Following Krebs and Sandow \cite{KrebsSandow1997}, the relations (\ref{equationsrule}) are consequence of subtracting (quadratic) operator relations of the type
\begin{equation*}\label{Prediff}
\Gamma_{\gamma\delta}^{\alpha\beta}D_{\alpha}D_{\beta}=D_{\gamma}X_{\delta}-X_{\gamma}D_{\delta},\ \ \text{for all}\ \ \gamma,\delta=0,1,\dotsc ,n-1,
\end{equation*}
 where $\Gamma_{\gamma\delta}^{\alpha\beta}\in \Bbbk$, and $D_{i}$'s and $X_{j}$'s are operators of a particular vector space such that not necessarily $[D_{i},X_{j}]=0$ holds \cite[p. 3168]{KrebsSandow1997}.
\end{enumerate}
\end{remark}

\section{Differential smoothness of algebras}\label{DefinitionsandpreliminariesDSA} 

We follow Brzezi\'nski and Sitarz's presentation on differential smoothness carried out in \cite[Section 2]{BrzezinskiSitarz2017} (c.f. \cite{Brzezinski2008, Brzezinski2014}).

\begin{definition}[{\cite[Section 2.1]{BrzezinskiSitarz2017}}]
\begin{enumerate}
    \item [\rm (i)] A {\em differential graded algebra} is a non-negatively graded algebra $\Omega$ with the product denoted by $\wedge$ together with a degree-one linear map 
    $$
    d:\Omega^{\bullet} \to \Omega^{\bullet +1}
    $$ 
    
    that satisfies the graded Leibniz rule and is such that $d \circ d = 0$. 
    
    \item [\rm (ii)] A differential graded algebra $(\Omega, d)$ is a {\em calculus over an algebra} $A$ if $\Omega^0 A = A$ and $\Omega^n A = A\ dA \wedge dA \wedge \dotsb \wedge dA$ ($dA$ appears $n$-times) for all $n\in \mathbb{N}$ (this last is called the {\em density condition}). We write $(\Omega A, d)$ with 
    $$
    \Omega A = \bigoplus_{n\in \mathbb{N}} \Omega^{n}A.
    $$
    
    By using the Leibniz rule, it follows that $\Omega^n A = dA \wedge dA \wedge \dotsb \wedge dA\ A$. A differential calculus $\Omega A$ is said to be {\em connected} if ${\rm ker}(d\mid_{\Omega^0 A}) = \Bbbk$.
    
    \item [\rm (iii)] A calculus $(\Omega A, d)$ is said to have {\em dimension} $n$ if $\Omega^n A\neq 0$ and $\Omega^m A = 0$ for all $m > n$. An $n$-dimensional calculus $\Omega A$ {\em admits a volume form} if $\Omega^n A$ is isomorphic to $A$ as a left and right $A$-module. 
\end{enumerate}
\end{definition}

The existence of a right $A$-module isomorphism means that there is a free generator, say $\omega$, of $\Omega^n A$ (as a right $A$-module), i.e. $\omega \in \Omega^n A$, such that all elements of $\Omega^n A$ can be uniquely expressed as $\omega a$ with $a \in A$. If $\omega$ is also a free generator of $\Omega^n A$ as a left $A$-module, this is said to be a {\em volume form} on $\Omega A$.

The right $A$-module isomorphism $\Omega^n A \to A$ corresponding to a volume form $\omega$ is denoted by $\pi_{\omega}$, i.e.
\begin{equation}\label{BrzezinskiSitarz2017(2.1)}
\pi_{\omega} (\omega a) = a, \quad {\rm for\ all}\ a\in A.
\end{equation}

By using that $\Omega^n A$ is also isomorphic to $A$ as a left $A$-module, any free generator $\omega $ induces an algebra endomorphism $\nu_{\omega}$ of $A$ by the formula
\begin{equation}\label{BrzezinskiSitarz2017(2.2)}
    a \omega = \omega \nu_{\omega} (a).
\end{equation}

Note that if $\omega$ is a volume form, then $\nu_{\omega}$ is an algebra automorphism.

Now, we proceed to recall the key ingredients of the {\em integral calculus} on $A$ as dual to its differential calculus. For more details, see Brzezinski et al. \cite{Brzezinski2008, BrzezinskiElKaoutitLomp2010}.

Let $(\Omega A, d)$ be a differential calculus on $A$. The space of $n$-forms $\Omega^n A$ is an $A$-bimodule. Consider $\mathcal{I}_{n}A$ the right dual of $\Omega^{n}A$, the space of all right $A$-linear maps $\Omega^{n}A\rightarrow A$, that is, $\mathcal{I}_{n}A := {\rm Hom}_{A}(\Omega^{n}(A),A)$. Notice that each of the $\mathcal{I}_{n}A$ is an $A$-bimodule with the actions
\begin{align*}
    (a\cdot\phi\cdot b)(\omega)=a\phi(b\omega),\quad {\rm for\ all}\ \phi \in \mathcal{I}_{n}A,\ \omega \in \Omega^{n}A\ {\rm and}\ a,b \in A.
\end{align*}

The direct sum of all the $\mathcal{I}_{n}A$, that is, $\mathcal{I}A = \bigoplus\limits_{n} \mathcal{I}_n A$, is a right $\Omega A$-module with action given by
\begin{align}\label{BrzezinskiSitarz2017(2.3)}
    (\phi\cdot\omega)(\omega')=\phi(\omega\wedge\omega'),\quad {\rm for\ all}\ \phi\in\mathcal{I}_{n + m}A, \ \omega\in \Omega^{n}A \ {\rm and} \ \omega' \in \Omega^{m}A.
\end{align}

\begin{definition}[{\cite[Definition 2.1]{Brzezinski2008}}]
A {\em divergence} (also called {\em hom-connection}) on $A$ is a linear map $\nabla: \mathcal{I}_1 A \to A$ such that
\begin{equation}\label{BrzezinskiSitarz2017(2.4)}
    \nabla(\phi \cdot a) = \nabla(\phi) a + \phi(da), \quad {\rm for\ all}\ \phi \in \mathcal{I}_1 A \ {\rm and} \ a \in A.
\end{equation}  
\end{definition}

Note that a divergence can be extended to the whole of $\mathcal{I}A$, 
\[
\nabla_n: \mathcal{I}_{n+1} A \to \mathcal{I}_{n} A,
\]

by considering
\begin{equation}\label{BrzezinskiSitarz2017(2.5)}
\nabla_n(\phi)(\omega) = \nabla(\phi \cdot \omega) + (-1)^{n+1} \phi(d \omega), \quad {\rm for\ all}\ \phi \in \mathcal{I}_{n+1}(A)\ {\rm and} \ \omega \in \Omega^n A.
\end{equation}

By putting together (\ref{BrzezinskiSitarz2017(2.4)}) and (\ref{BrzezinskiSitarz2017(2.5)}), we get the Leibniz rule 
\begin{equation}
    \nabla_n(\phi \cdot \omega) = \nabla_{m + n}(\phi) \cdot \omega + (-1)^{m + n} \phi \cdot d\omega,
\end{equation}

for all elements $\phi \in \mathcal{I}_{m + n + 1} A$ and $\omega \in \Omega^m A$ \cite[Lemma 3.2]{Brzezinski2008}. In the case $n = 0$, if ${\rm Hom}_A(A, M)$ is canonically identified with $M$, then $\nabla_0$ reduces to the classical Leibniz rule.

\begin{definition}[{\cite[Definition 3.4]{Brzezinski2008}}]
The right $A$-module map 
$$
F = \nabla_0 \circ \nabla_1: {\rm Hom}_A(\Omega^{2} A, M) \to M
$$ is called a {\em curvature} of a hom-connection $(M, \nabla_0)$. $(M, \nabla_0)$ is said to be {\em flat} if its curvature is the zero map, that is, if $\nabla \circ \nabla_1 = 0$. This condition implies that $\nabla_n \circ \nabla_{n+1} = 0$ for all $n\in \mathbb{N}$.
\end{definition}

$\mathcal{I} A$ together with the $\nabla_n$ form a chain complex called the {\em complex of integral forms} over $A$. The cokernel map of $\nabla$, that is, 
$$
\Lambda: A \to {\rm Coker} \nabla = A / {\rm Im} \nabla
$$ 

is said to be the {\em integral on $A$ associated to} $\mathcal{I}A$.

Given a left $A$-module $X$ with action $a\cdot x$, for all $a\in A,\ x \in X$, and an algebra automorphism $\nu$ of $A$, the notation $^{\nu}X$ stands for $X$ with the $A$-module structure twisted by $\nu$, i.e. with the $A$-action $a\otimes x \mapsto \nu(a)\cdot x $.

The following definition of an \textit{integrable differential calculus} seeks to portray a version of Hodge star isomorphisms between the complex of differential forms of a differentiable manifold and a complex of dual modules of it \cite[p. 112]{Brzezinski2015}. 

\begin{definition}[{\cite[Definition 2.1]{BrzezinskiSitarz2017}}]
An $n$-dimensional differential calculus $(\Omega A, d)$ is said to be {\em integrable} if $(\Omega A, d)$ admits a complex of integral forms $(\mathcal{I}A, \nabla)$ for which there exist an algebra automorphism $\nu$ of $A$ and $A$-bimodule isomorphisms 
$$
\Theta_k: \Omega^{k} A \to ^{\nu} \mathcal{I}_{n-k}A, \quad k = 0, \dotsc, n
$$

rendering commmutative the following diagram:
\[
\begin{tikzcd}
A \arrow{r}{d} \arrow{d}{\Theta_0} & \Omega^{1} A \arrow{d}{\Theta_1} \arrow{r}{d} & \Omega^2 A  \arrow{d}{\Theta_2} \arrow{r}{d} & \dotsb \arrow{r}{d} & \Omega^{n-1} A \arrow{d}{\Theta_{n-1}} \arrow{r}{d} & \Omega^n A  \arrow{d}{\Theta_n} \\ ^{\nu} \mathcal{I}_n A \arrow[swap]{r}{\nabla_{n-1}} & ^{\nu} \mathcal{I}_{n-1} A \arrow[swap]{r}{\nabla_{n-2}} & ^{\nu} \mathcal{I}_{n-2} A \arrow[swap]{r}{\nabla_{n-3}} & \dotsb \arrow[swap]{r}{\nabla_{1}} & ^{\nu} \mathcal{I}_{1} A \arrow[swap]{r}{\nabla} & ^{\nu} A
\end{tikzcd}
\]

The $n$-form $\omega:= \Theta_n^{-1}(1)\in \Omega^n A$ is called an {\em integrating volume form}. 
\end{definition}

The algebra of complex matrices $M_n(\mathbb{C})$ with the $n$-dimensional calculus generated by derivations presented by Dubois-Violette et al. \cite{DuboisViolette1988, DuboisVioletteKernerMadore1990}, the quantum group $SU_q(2)$ with the three-dimensional left covariant calculus developed by Woronowicz \cite{Woronowicz1987} and the quantum standard sphere with the restriction of the above calculus, are examples of algebras admitting integrable calculi. For more details on the subject, see Brzezi\'nski et al. \cite{BrzezinskiElKaoutitLomp2010}. 

The following proposition shows that the integrability of a differential calculus can be defined without explicit reference to integral forms. This allows us to guarantee the integrability by considering the existence of finitely generator elements that allow to determine left and right components of any homogeneous element of $\Omega(A)$.

\begin{proposition}[{\cite[Theorem 2.2]{BrzezinskiSitarz2017}}]\label{integrableequiva} 
Let $(\Omega A, d)$ be an $n$-dimensional differential calculus over an algebra $A$. The following assertions are equivalent:
\begin{enumerate}
    \item [\rm (1)] $(\Omega A, d)$ is an integrable differential calculus.
    
    \item [\rm (2)] There exists an algebra automorphism $\nu$ of $A$ and $A$-bimodule isomorphisms 
    $$
    \Theta_k : \Omega^k A \rightarrow \ ^{\nu}\mathcal{I}_{n-k}A, \quad k = 0, \ldots, n
    $$

    such that, for all $\omega'\in \Omega^k A$ and $\omega''\in \Omega^mA$,
    \begin{align*}
        \Theta_{k+m}(\omega'\wedge\omega'')=(-1)^{(n-1)m}\Theta_k(\omega')\cdot\omega''.
    \end{align*}
    
    \item [\rm (3)] There exists an algebra automorphism $\nu$ of $A$ and an $A$-bimodule map $\vartheta:\Omega^nA\rightarrow\ ^{\nu}A$ such that all left multiplication maps
    \begin{align*}
    \ell_{\vartheta}^{k}:\Omega^k A &\ \rightarrow \mathcal{I}_{n-k}A, \\
    \omega' &\ \mapsto \vartheta\cdot\omega', \quad k = 0, 1, \dotsc, n,
    \end{align*}
    where the actions $\cdot$ are defined by {\rm (}\ref{BrzezinskiSitarz2017(2.3)}{\rm )}, are bijective.
    
    \item [\rm (4)] $(\Omega A, d)$ has a volume form $\omega$ such that all left multiplication maps
    \begin{align*}
        \ell_{\pi_{\omega}}^{k}:\Omega^k A &\ \rightarrow \mathcal{I}_{n-k}A, \\
        \omega' &\ \mapsto \pi_{\omega} \cdot \omega', \quad k=0,1, \dotsc, n-1,
    \end{align*}
    
    where $\pi_{\omega}$ is defined by {\rm (}\ref{BrzezinskiSitarz2017(2.1)}{\rm )}, are bijective.
\end{enumerate}
\end{proposition}

A volume form $\omega\in \Omega^nA$ is an {\em integrating form} if and only if it satisfies condition $(4)$ of Proposition \ref{integrableequiva} \cite[Remark 2.3]{BrzezinskiSitarz2017}.

The most interesting cases of differential calculi are those where $\Omega^k A$ are finitely generated and projective right or left (or both) $A$-modules \cite{Brzezinski2011}.

\begin{proposition}\label{BrzezinskiSitarz2017Lemmas2.6and2.7}
\begin{enumerate}
\item [\rm (1)] \cite[Lemma 2.6]{BrzezinskiSitarz2017} Consider $(\Omega A, d)$ an integrable and $n$-dimensional calculus over $A$ with integrating form $\omega$. Then $\Omega^{k} A$ is a finitely generated projective right $A$-module if there exist a finite number of forms $\omega_i \in \Omega^{k} A$ and $\overline{\omega}_i \in \Omega^{n-k} A$ such that, for all $\omega' \in \Omega^{k} A$, we have that 
\begin{equation*}
\omega' = \sum_{i} \omega_i \pi_{\omega} (\overline{\omega}_i \wedge \omega').
\end{equation*}

\item [\rm (2)] \cite[Lemma 2.7]{BrzezinskiSitarz2017} Let $(\Omega A, d)$ be an $n$-dimensional calculus over $A$ admitting a volume form $\omega$. Assume that for all $k = 1, \ldots, n-1$, there exists a finite number of forms $\omega_{i}^{k},\overline{\omega}_{i}^{k} \in \Omega^{k}(A)$ such that for all $\omega'\in \Omega^kA$, we have that
\begin{equation*}
\omega'=\displaystyle\sum_i\omega_{i}^{k}\pi_\omega(\overline{\omega}_{i}^{n-k}\wedge\omega')=\displaystyle\sum_i\nu_{\omega}^{-1}(\pi_\omega(\omega'\wedge\omega_{i}^{n-k}))\overline{\omega}_{i}^{k},
\end{equation*}

where $\pi_{\omega}$ and $\nu_{\omega}$ are defined by {\rm (}\ref{BrzezinskiSitarz2017(2.1)}{\rm )} and {\rm (}\ref{BrzezinskiSitarz2017(2.2)}{\rm )}, respectively. Then $\omega$ is an integral form and all the $\Omega^{k}A$ are finitely generated and projective as left and right $A$-modules.
\end{enumerate}
\end{proposition}

Brzezi\'nski and Sitarz \cite[p. 421]{BrzezinskiSitarz2017} asserted that to connect the integrability of the differential graded algebra $(\Omega A, d)$ with the algebra $A$, it is necessary to relate the dimension of the differential calculus $\Omega A$ with that of $A$, and since we are dealing with algebras that are deformations of coordinate algebras of affine varieties, the {\em Gelfand-Kirillov dimension} introduced by Gelfand and Kirillov \cite{GelfandKirillov1966, GelfandKirillov1966b} seems to be the best suited. Briefly, given an affine $\Bbbk$-algebra $A$, the {\em Gelfand-Kirillov dimension of} $A$, denoted by ${\rm GKdim}(A)$, is given by
\[
{\rm GKdim}(A) := \underset{n\to \infty}{\rm lim\ sup} \frac{{\rm log}({\rm dim}\ V^{n})}{{\rm log}\ n},
\]

where $V$ is a finite-dimensional subspace of $A$ that generates $A$ as an algebra. This definition is independent of choice of $V$. If $A$ is not affine, then its Gelfand-Kirillov dimension is defined to be the supremum of the Gelfand-Kirillov dimensions of all affine subalgebras of $A$. An affine domain of Gelfand-Kirillov dimension zero is precisely a division ring that is finite-dimensional over its center. In the case of an affine domain of Gelfand-Kirillov dimension one over $\Bbbk$, this is precisely a finite module over its center, and thus polynomial identity. In some sense, this dimensions measures the deviation of the algebra $A$ from finite dimensionality. For more details about this dimension, see the excellent treatment developed by Krause and Lenagan \cite{KrauseLenagan2000}.

After preliminaries above, we arrive to the key notion of this paper.

\begin{definition}[{\cite[Definition 2.4]{BrzezinskiSitarz2017}}]\label{BrzezinskiSitarz2017Definition2.4}
An affine algebra $A$ with integer Gelfand-Kirillov dimension $n$ is said to be {\em differentially smooth} if it admits an $n$-dimensional connected integrable differential calculus $(\Omega A, d)$.
\end{definition}

From Definition \ref{BrzezinskiSitarz2017Definition2.4} it follows that a differentially smooth algebra comes equipped with a well-behaved differential structure and with the precise concept of integration \cite[p. 2414]{BrzezinskiLomp2018}.

As we said in the Introduction, different authors have characterized the differential smoothness of several noncommutative algebras \cite{Brzezinski2015, Brzezinski2016, BrzezinskiElKaoutitLomp2010, BrzezinskiLomp2018, BrzezinskiSitarz2017, DuboisVioletteKernerMadore1990, Karacuha2015, KaracuhaLomp2014, ReyesSarmiento2022, RubianoReyes2024DSBiquadraticAlgebras, RubianoReyes2024DSDoubleOreExtensions}. 

\begin{remark}
There are examples of algebras that are not differentially smooth. Consider the commutative algebra $A = \mathbb{C}[x, y] / \langle xy \rangle$. A proof by contradiction shows that for this algebra there are no one-dimensional connected integrable calculi over $A$, so it cannot be differentially smooth \cite[Example 2.5]{BrzezinskiSitarz2017}.
\end{remark}

\section{Differential smoothness of diffusion algebras}\label{DSDANgenerators}

Throughout this section, $\mathcal{D}$ denotes a diffusion algebra on $N\geq 3$ generators. 

\subsection{Diffusion algebras on four generators}\label{ClasDF4}

For the four-generator construction, the explicit algebras are given in Tables \ref{Diffusion4} and \ref{Diffusion4(2)}. It is worth noting that several of these algebras appear to be isomorphic upon some index change. However, isomorphism would have to be shown explicitly, and what is desired here is to obtain all possible combinations of algebras, even if some are isomorphic.

\begin{table}[h]
\caption{Diffusion algebras on four generators}
\label{Diffusion4}
\begin{center}
\resizebox{12.6cm}{!}{
\setlength\extrarowheight{6pt}
\begin{tabular}{ |c|c|c|c| } 
\hline
{\rm Diffusion algebra} & {\rm Relations} & {\rm Conditions} & {\rm Restrictions} \\
\hline
\multirow{12}{*}{$A_{I}$} & $gD_2D_1-gD_1D_2 = x_2D_1-x_1D_2, \quad gD_3D_1-gD_1D_3 = x_3D_1-x_1D_3$, & \multirow{3}{*}{$I=\{1, 2, 3\}$, $S=\{4\}$} & \multirow{3}{*}{$g, g_4 \ne 0$} \\
& $gD_3D_2-gD_2D_3 = x_3D_2-x_2D_3, \quad g_4D_4D_1-g_4D_1D_4=x_1D_4$,  & & \\
& $g_4D_4D_2-g_4D_2D_4=x_2D_4, \quad g_4D_4D_3-g_4D_3D_4=x_1D_4$ & & \\ \cline{2-4}
 & $gD_3D_1-gD_1D_3 = x_3D_1-x_1D_3, \quad gD_4D_1-gD_1D_4 = x_4D_1-x_1D_4$, & \multirow{3}{*}{$I=\{1, 3, 4\}$, $S=\{2\}$} & \multirow{3}{*}{$g, g_2 \ne 0$} \\
& $gD_4D_3-gD_3D_4 = x_4D_3-x_3D_4, \quad g_2D_2D_1-g_2D_1D_2=x_1D_2$,  & & \\
& $g_2D_3D_2-g_2D_2D_3=x_3D_2, \quad g_2D_4D_2-g_2D_2D_4=x_4D_2$ & & \\ \cline{2-4}
 & $gD_2D_1-gD_1D_2 = x_2D_1-x_1D_2, \quad gD_4D_1-gD_1D_4 = x_4D_1-x_1D_4$, & \multirow{3}{*}{$I=\{1, 2, 4\}$, $S=\{3\}$} & \multirow{3}{*}{$g, g_3 \ne 0$} \\
& $gD_4D_3-gD_3D_4 = x_4D_3-x_3D_4, \quad g_3D_3D_1-g_3D_1D_3=x_1D_3$,  & & \\
& $g_3D_3D_2-g_3D_2D_3=x_2D_3, \quad g_3D_3D_4-g_3D_4D_3=x_4D_3$ & & \\ \cline{2-4}
 & $gD_3D_2-gD_2D_3 = x_3D_2-x_2D_3, \quad gD_4D_2-gD_2D_4 = x_4D_2-x_2D_4$, & \multirow{3}{*}{$I=\{2, 3, 4\}$, $S=\{1\}$} & \multirow{3}{*}{$g, g_1 \ne 0$} \\
& $gD_4D_3-gD_3D_4 = x_4D_3-x_3D_4, \quad g_1D_2D_1-g_1D_1D_2=x_2D_1$,  & & \\
& $g_1D_3D_1-g_1D_1D_3=x_3D_1, \quad g_1D_4D_1-g_1D_1D_4=x_4D_1$ & & \\
\hline
\multirow{12}{*}{$A_{II}$} & $(g_1-g_2)D_1D_2= x_2D_1-x_1D_2, \quad (g_1-g_3)D_1D_3= x_3D_1-x_1D_3$, & \multirow{3}{*}{$I=\{1, 2, 3, 4\}$} &  \multirow{3}{*}{$g_i \ne g_j$ for $i, j \in I$, $i \ne j$}  \\ 
&  $(g_1-g_4)D_1D_4= x_4D_1-x_1D_4, \quad (g_2-g_3)D_2D_3= x_3D_2-x_2D_3$, & & \\
&  $(g_2-g_4)D_2D_4= x_4D_2-x_2D_4, \quad (g_3-g_4)D_3D_4= x_4D_3-x_3D_4$ & & \\ \cline{2-4}
& $(g_1-g_2)D_1D_2= x_2D_1-x_1D_2, \quad (g_1-g_3)D_1D_3= x_3D_1-x_1D_3$, & \multirow{3}{*}{$I = \{1, 2, 3\} $, $T_a^{\circ}=\{4 \}$}  &  \multirow{3}{*}{$g_i \ne g_j$ for $i, j \in I$, $i \ne j$, $g_i \ne g_1^{\circ}$} \\ 
&  $(g_2-g_3)D_3D_2 = x_3D_2-x_2D_3, \quad (g_1+g_1^{\circ})D_1D_4=-x_1D_4$, & & \\
& $(g_2+g_1^{\circ})D_2D_4=-x_2D_4, \quad (g_3+g_1^{\circ})D_3D_4=-x_3D_4$ & & \\ \cline{2-4}
& $(g_2-g_3)D_2D_3= x_3D_2-x_2D_3, \quad (g_2-g_4)D_2D_4= x_4D_2-x_2D_4$, & \multirow{3}{*}{$I = \{2, 3, 4\} $, $T_a^{\circ}=\{1 \}$}  &  \multirow{3}{*}{$g_i \ne g_j$ for $i, j \in I$, $i \ne j$, $g_i \ne g_1^{\circ}$} \\ 
&  $(g_3-g_4)D_4D_3 = x_4D_3-x_3D_4, \quad (g_2+g_1^{\circ})D_1D_2=-x_2D_1$, & & \\
& $(g_3+g_1^{\circ})D_1D_3=-x_3D_1, \quad (g_4+g_1^{\circ})D_1D_4=-x_4D_1$ & & \\ \cline{2-4}
& $(g_1-g_3)D_1D_3= x_3D_1-x_1D_3, \quad (g_1-g_4)D_1D_4= x_4D_1-x_1D_4$, & \multirow{3}{*}{$I = \{1, 3, 4\} $, $T_b^{\bullet}=\{2\}$}  &  \multirow{3}{*}{$g_i \ne g_j$ for $i, j \in I$, $i \ne j$, $g_i \ne \mp g_1^{\pm}$} \\ 
&  $(g_3-g_4)D_4D_3 = x_4D_3-x_3D_4, \quad (g_1+g_1^{+})D_1D_2=-x_1D_2$, & & \\
& $(g_1^{-}-g_3)D_2D_3=x_3D_2, \quad (g_1^{-}-g_4)D_2D_4=x_4D_2$ & & \\ \cline{2-4}
& $(g_1-g_2)D_1D_2= x_2D_1-x_1D_2, \quad (g_1-g_4)D_1D_4= x_4D_1-x_1D_4$, & \multirow{3}{*}{$I = \{1, 2, 4\} $, $T_b^{\bullet}=\{3\}$}  &  \multirow{3}{*}{$g_i \ne g_j$ for $i, j \in I$, $i \ne j$, $g_i \ne \mp g_1^{\pm}$} \\ 
&  $(g_2-g_4)D_2D_4 = x_4D_2-x_2D_4, \quad (g_1+g_1^{+})D_1D_3=-x_1D_3$, & & \\
& $(g_2+g_1^{+})D_2D_3=-x_2D_3, \quad (g_1^{-}-g_4)D_3D_4=x_4D_3$ & & \\ 
\hline
\multirow{37}{*}{$B$} & $gD_1D_2-(g-\Lambda)D_2D_1=x_2D_1-x_1D_2, \quad g_3D_1D_3-(g_3-\Lambda)D_3D_1=-x_1D_3$, & \multirow{3}{*}{$I=\{1, 2\}$, $S=\{3, 4\}$} & \multirow{3}{*}{$g \ne 0$, $g_s\not \in \{0, \Lambda\}$, $s\in S$} \\ 
& $g_4D_1D_4-(g_4-\Lambda)D_4D_1=-x_1D_4, \quad g_3D_2D_3-(g_3-\Lambda)D_3D_2=-x_2D_3$, & & \\
& $g_4D_2D_4-(g_4-\Lambda)D_4D_2=-x_2D_4$& & \\ \cline{2-4}
 & $gD_1D_3-(g-\Lambda)D_3D_1=x_3D_1-x_1D_3, \quad g_2D_1D_2-(g_2-\Lambda)D_2D_1=-x_1D_2$, & \multirow{3}{*}{$I=\{1, 3\}$, $S=\{2, 4\}$} & \multirow{3}{*}{$g \ne 0$, $g_2\ne 0$, $g_4 \not \in  \{0, \Lambda \}$} \\ 
& $g_4D_1D_4-(g_4-\Lambda)D_4D_1=-x_1D_4, \quad g_2D_3D_2-(g_2-\Lambda)D_2D_3=-x_3D_2$, & & \\
& $g_4D_3D_4-(g_4-\Lambda)D_4D_3=-x_3D_4$& & \\ \cline{2-4}
& $gD_1D_4-(g-\Lambda)D_4D_1=x_4D_1-x_1D_4, \quad g_3D_1D_3-(g_3-\Lambda)D_3D_1=-x_1D_3$, & \multirow{3}{*}{$I=\{1, 4\}$, $S=\{2, 3\}$} & \multirow{3}{*}{$g \ne 0$, $g_s\ne 0$, $s\in S$} \\ 
& $g_2D_1D_2-(g_2-\Lambda)D_2D_1=-x_1D_2, \quad g_3D_4D_3-(g_3-\Lambda)D_3D_4=-x_4D_3$, & & \\
& $g_2D_4D_2-(g_2-\Lambda)D_2D_4=-x_4D_2$& & \\ \cline{2-4}
 & $gD_2D_3-(g-\Lambda)D_3D_2=x_3D_2-x_2D_3, \quad g_1D_2D_1-(g_1-\Lambda)D_1D_2=-x_2D_1$, & \multirow{3}{*}{$I=\{2, 3\}$, $S=\{1, 4\}$} & \multirow{3}{*}{$g \ne 0$, $g_s\not \in  \{ 0, \Lambda \}$, $s\in S$} \\ 
& $g_4D_2D_4-(g_4-\Lambda)D_4D_2=-x_2D_4, \quad g_1D_3D_1-(g_1-\Lambda)D_1D_3=-x_3D_1$, & & \\
& $g_4D_3D_4-(g_4-\Lambda)D_4D_3=-x_3D_4$& & \\ \cline{2-4}
& $gD_2D_4-(g-\Lambda)D_4D_2=x_4D_2-x_2D_4, \quad g_1D_2D_1-(g_1-\Lambda)D_1D_2=-x_2D_1$, & \multirow{3}{*}{$I=\{2, 4\}$, $S=\{1, 3\}$} & \multirow{3}{*}{$g \ne 0$, $g_3\ne 0$, $g_1 \not \in \{0, \Lambda\}$} \\ 
& $g_3D_2D_3-(g_3-\Lambda)D_3D_2=-x_2D_3, \quad g_1D_4D_1-(g_1-\Lambda)D_1D_4=-x_4D_1$, & & \\
& $g_3D_4D_3-(g_3-\Lambda)D_3D_4=-x_4D_3$& & \\ \cline{2-4}
& $gD_3D_4-(g-\Lambda)D_4D_3=x_4D_3-x_3D_4, \quad g_1D_3D_1-(g_1-\Lambda)D_1D_3=-x_3D_1$, & \multirow{3}{*}{$I=\{3, 4\}$, $S=\{1, 2\}$} & \multirow{3}{*}{$g \ne 0$, $g_s\not \in \{0, \Lambda \}$, $s\in S$} \\ 
& $g_2D_3D_2-(g_2-\Lambda)D_2D_3=-x_3D_2, \quad g_1D_4D_1-(g_1-\Lambda)D_1D_4=-x_4D_1$, & & \\
& $g_2D_4D_2-(g_2-\Lambda)D_2D_4=-x_4D_2$& & \\ \cline{2-4}
& $gD_1D_2-(g-\Lambda)D_2D_1 = x_2D_1-x_1D_2, \quad g_3D_1D_3-(g_3-\Lambda)D_3D_1=-x_1D_3$, &  \multirow{3}{*}{$I=\{1, 2\}$, $S=\{3\}$, $T_a^{\circ}=\{4\}$} & \multirow{3}{*}{$g \ne 0$, $g_3, g_a^{\circ} \not \in \{0, \Lambda\}$} \\ 
& $g_3D_2D_3-(g_3-\Lambda)D_3D_2=-x_2D_3, \quad g_a^{\circ}D_1D_4=-x_1D_4$, &  & \\ 
& $(g_a^{\circ}-\Lambda)D_2D_4=-x_2D_4$, &  & \\ \cline{2-4}
& $gD_1D_2-(g-\Lambda)D_2D_1 = x_2D_1-x_1D_2, \quad g_4D_1D_4-(g_4-\Lambda)D_4D_1=-x_1D_4$, & 
\multirow{3}{*}{$I=\{1, 2\}$, $S=\{4\}$, $T_a^{\circ}=\{3\}$} & \multirow{3}{*}{$g \ne 0$, $g_4, g_a^{\circ} \not \in \{0, \Lambda\}$} \\ 
& $g_4D_2D_4-(g_4-\Lambda)D_4D_2=-x_2D_4, \quad g_a^{\circ}D_1D_3=-x_1D_3$, &  & \\ 
& $(g_a^{\circ}-\Lambda)D_2D_3=-x_2D_3$, &  & \\ \cline{2-4}
& $gD_1D_3-(g-\Lambda)D_3D_1 = x_3D_1-x_1D_3, \quad g_2D_1D_2-(g_2-\Lambda)D_2D_1=-x_1D_2$, & 
\multirow{3}{*}{$I=\{1, 3\}$, $S=\{2\}$, $T_a^{\circ}=\{4\}$} & \multirow{3}{*}{$g \ne 0$, $g_2 \ne 0, g_a^{\circ} \not \in \{0, \Lambda\}$} \\ 
& $g_2D_3D_2-(g_2-\Lambda)D_2D_3=-x_3D_2, \quad g_a^{\circ}D_1D_4=-x_1D_4$, &  & \\ 
& $(g_a^{\circ}-\Lambda)D_3D_4=-x_3D_4$, &  & \\ \cline{2-4}
& $gD_2D_3-(g-\Lambda)D_3D_2 = x_3D_2-x_2D_3, \quad g_1D_2D_1-(g_1-\Lambda)D_1D_2=-x_2D_1$, & 
\multirow{3}{*}{$I=\{2, 3\}$, $S=\{1\}$, $T_a^{\circ}=\{4\}$} & \multirow{3}{*}{$g \ne 0$, $g_1, g_a^{\circ} \not \in \{0, \Lambda\}$} \\ 
& $g_1D_3D_1-(g_1-\Lambda)D_1D_3=-x_3D_1, \quad g_a^{\circ}D_2D_4=-x_2D_4$, &  & \\ 
& $(g_a^{\circ}-\Lambda)D_3D_4=-x_3D_4$, &  & \\ \cline{2-4}
& $gD_2D_3-(g-\Lambda)D_3D_2 = x_3D_2-x_2D_3, \quad g_4D_2D_4-(g_4-\Lambda)D_4D_2=-x_2D_4$, & 
\multirow{3}{*}{$I=\{2, 3\}$, $S=\{4\}$, $T_a^{\circ}=\{1\}$} & \multirow{3}{*}{$g \ne 0$, $g_4, g_a^{\circ} \not \in \{0, \Lambda\}$} \\ 
& $g_4D_3D_4-(g_4-\Lambda)D_4D_3=-x_3D_4, \quad g_a^{\circ}D_1D_2=-x_2D_1$, &  & \\ 
& $(g_a^{\circ}-\Lambda)D_1D_3=-x_3D_1$, &  & \\ \cline{2-4}
& $gD_2D_4-(g-\Lambda)D_4D_2 = x_4D_2-x_2D_4, \quad g_3D_2D_3-(g_3-\Lambda)D_3D_2=-x_2D_3$, & 
\multirow{3}{*}{$I=\{2, 4\}$, $S=\{3\}$, $T_a^{\circ}=\{1\}$} & \multirow{3}{*}{$g \ne 0$, $g_3 \ne 0, g_a^{\circ} \not \in \{0, \Lambda\}$} \\ 
& $g_3D_4D_3-(g_3-\Lambda)D_3D_4=-x_4D_3, \quad g_a^{\circ}D_1D_2=-x_2D_1$, &  & \\ 
& $(g_a^{\circ}-\Lambda)D_1D_4=-x_4D_1$, &  & \\ \cline{2-4}
& $gD_3D_4-(g-\Lambda)D_4D_3 = x_4D_3-x_3D_4, \quad g_1D_3D_1-(g_1-\Lambda)D_1D_3=-x_3D_1$, & 
\multirow{3}{*}{$I=\{3, 4\}$, $S=\{1\}$, $T_a^{\circ}=\{2\}$} & \multirow{3}{*}{$g \ne 0$, $g_1, g_a^{\circ} \not \in \{0, \Lambda\}$} \\ 
& $g_1D_4D_1-(g_1-\Lambda)D_1D_4=-x_4D_1, \quad g_a^{\circ}D_2D_3=-x_3D_2$, &  & \\ 
& $(g_a^{\circ}-\Lambda)D_2D_4=-x_4D_2$, &  & \\ \cline{2-4}
& $gD_3D_4-(g-\Lambda)D_4D_3 = x_4D_3-x_3D_4, \quad g_2D_3D_2-(g_2-\Lambda)D_2D_3=-x_3D_2$, & 
\multirow{3}{*}{$I=\{3, 4\}$, $S=\{2\}$, $T_a^{\circ}=\{1\}$} & \multirow{3}{*}{$g \ne 0$, $g_2, g_a^{\circ} \not \in \{0, \Lambda\}$} \\ 
& $g_2D_4D_2-(g_2-\Lambda)D_2D_4=-x_4D_2, \quad g_a^{\circ}D_1D_3=-x_3D_1$, &  & \\ 
& $(g_a^{\circ}-\Lambda)D_1D_4=-x_4D_1$, &  & \\ \cline{2-4}
& $gD_1D_3-(g-\Lambda)D_3D_1 = x_3D_1-x_1D_3, \quad g_4D_1D_4-(g_4-\Lambda)D_4D_1=-x_1D_4$, & 
\multirow{3}{*}{$I=\{1, 3\}$, $S=\{4\}$, $T_b^{\bullet}=\{2\}$} & \multirow{3}{*}{$g \ne 0$, $g_4 \not \in \{0, \Lambda\}$, $g_b^{\pm} \ne 0$} \\ 
& $g_4D_3D_4-(g_4-\Lambda)D_4D_3=-x_3D_4, \quad g_b^{+}D_1D_2=-x_1D_2$, &  & \\ 
& $g_b^{-}D_2D_3=x_2D_3$, &  & \\ 
\hline
\end{tabular}
}
\end{center}
\end{table}

\begin{table}[h]
\caption{Diffusion algebras on four generators}
\label{Diffusion4(2)}
\begin{center}
\resizebox{12.6cm}{!}{
\setlength\extrarowheight{6pt}
\begin{tabular}{ |c|c|c|c| } 
\hline
{\rm Diffusion algebra} & {\rm Relations} & {\rm Conditions} & {\rm Restrictions} \\
\hline
\multirow{27}{*}{$B$} & $gD_1D_4-(g-\Lambda)D_4D_1 = x_4D_1-x_1D_4, \quad g_2D_1D_2-(g_2-\Lambda)D_2D_1=-x_1D_2$, & 
\multirow{3}{*}{$I=\{1, 4\}$, $S=\{2\}$, $T_b^{\bullet}=\{3\}$} & \multirow{3}{*}{$g \ne 0$, $g_2 \ne 0$, $g_b^{\pm} \ne 0$} \\ 
& $g_2D_4D_2-(g_2-\Lambda)D_2D_4=-x_4D_2, \quad g_b^{+}D_1D_3=-x_1D_3$, &  & \\ 
& $g_b^{-}D_3D_4=x_4D_3$, &  & \\ \cline{2-4}
& $gD_1D_4-(g-\Lambda)D_4D_1 = x_4D_1-x_1D_4, \quad g_3D_1D_3-(g_3-\Lambda)D_3D_1=-x_1D_3$, & 
\multirow{3}{*}{$I=\{1, 4\}$, $S=\{3\}$, $T_b^{\bullet}=\{2\}$} & \multirow{3}{*}{$g \ne 0$, $g_3\ne 0$, $g_b^{\pm} \ne 0$} \\ 
& $g_3D_4D_3-(g_3-\Lambda)D_3D_4=-x_4D_3, \quad g_b^{+}D_1D_2=-x_1D_2$, &  & \\ 
& $g_b^{-}D_2D_4=x_4D_2$, &  & \\ \cline{2-4}
& $gD_2D_4-(g-\Lambda)D_4D_2 = x_4D_2-x_2D_4, \quad g_1D_2D_1-(g_1-\Lambda)D_1D_2=-x_2D_1$, & 
\multirow{3}{*}{$I=\{2, 4\}$, $S=\{1\}$, $T_b^{\bullet}=\{3\}$} & \multirow{3}{*}{$g \ne 0$, $g_1 \not \in \{0, \Lambda\}$, $g_b^{\pm} \ne 0$} \\ 
& $g_1D_4D_1-(g_1-\Lambda)D_1D_4=-x_4D_1, \quad g_b^{+}D_2D_3=-x_2D_3$, &  & \\ 
& $g_b^{-}D_3D_4=x_4D_3$, &  & \\ \cline{2-4}
& $gD_1D_2-(g-\Lambda)D_2D_1 = x_1D_2-x_2D_1, \quad g_a^{\circ}D_1D_3=-x_1D_3$, & 
\multirow{3}{*}{$I=\{1, 2\}$, $T_a^{\circ}=\{3, 4\}$} & \multirow{3}{*}{$g \ne 0$, $g_a^{\circ} \not \in \{0, \Lambda\}$} \\ 
& $g_a^{\circ}D_1D_4=-x_1D_4, \quad (g_a^{\circ}-\Lambda)D_2D_3=-x_2D_3$, &  & \\ 
& $(g_a^{\circ}-\Lambda)D_2D_4=-x_2D_4$, &  & \\ \cline{2-4}
& $gD_1D_3-(g-\Lambda)D_3D_1 = x_1D_3-x_3D_1, \quad g_a^{\circ}D_1D_4=-x_1D_4$, & 
\multirow{3}{*}{$I=\{1, 3\}$, $T_b^{\bullet}=\{ 2\}$, $T_a^{\circ}=\{ 4\}$} & \multirow{3}{*}{$g \ne 0$, $g_a^{\circ} \not \in \{0, \Lambda\}$, $g_b^{\pm}\ne 0$} \\ 
& $(g_a^{\circ}-\Lambda)D_3D_4=-x_3D_4, \quad g_b^{+}D_1D_2=-x_1D_2$, &  & \\ 
& $g_b^{-}D_2D_3=x_3D_2$, &  & \\ \cline{2-4}
& $gD_1D_4-(g-\Lambda)D_4D_1 = x_1D_4-x_4D_1, \quad g_b^{+}D_1D_2=-x_1D_2$, & 
\multirow{3}{*}{$I=\{1, 4\}$, $T_b^{\bullet}=\{ 2, 3\}$} & \multirow{3}{*}{$g \ne 0$, $g_b^{\pm}\ne 0$} \\ 
& $g_b^{-}D_2D_4=x_4D_2, \quad g_b^{+}D_1D_3=-x_1D_3$, &  & \\ 
& $g_b^{-}D_3D_4=x_4D_3$, &  & \\ \cline{2-4}
& $gD_2D_3-(g-\Lambda)D_3D_2 = x_2D_3-x_3D_2, \quad g_a^{\circ}D_1D_2=-x_1D_2$, & 
\multirow{3}{*}{$I=\{2, 3\}$, $T_a^{\circ}=\{1, 4\}$} & \multirow{3}{*}{$g \ne 0$, $g_a^{\circ} \not \in \{0, \Lambda\}$} \\ 
& $g_a^{\circ}D_1D_3=-x_1D_3, \quad (g_a^{\circ}-\Lambda)D_1D_3=-x_3D_1$, &  & \\ 
& $(g_a^{\circ}-\Lambda)D_3D_4=-x_3D_4$, &  & \\ \cline{2-4}
& $gD_2D_4-(g-\Lambda)D_4D_2 = x_2D_4-x_4D_2, \quad g_a^{\circ}D_1D_2=-x_2D_1$, & 
\multirow{3}{*}{$I=\{2, 4\}$, $T_a^{\circ}=\{ 1\}$, $T_b^{\bullet}=\{ 3\}$} & \multirow{3}{*}{$g \ne 0$, $g_a^{\circ} \not \in \{0, \Lambda\}$, $g_b^{\pm}\ne 0$} \\ 
& $(g_a^{\circ}-\Lambda)D_1D_4=-x_4D_1, \quad g_b^{+}D_2D_3=-x_2D_3$, &  & \\ 
& $g_b^{-}D_3D_4=x_4D_3$, &  & \\ \cline{2-4}
& $gD_3D_4-(g-\Lambda)D_4D_3 = x_3D_4-x_4D_3, \quad g_a^{\circ}D_1D_3=-x_3D_1$, & 
\multirow{3}{*}{$I=\{3, 4\}$, $T_a^{\circ}=\{1, 2\}$} & \multirow{3}{*}{$g \ne 0$, $g_a^{\circ} \not \in \{0, \Lambda\}$} \\ 
& $g_a^{\circ}D_2D_3=-x_3D_2, \quad (g_a^{\circ}-\Lambda)D_1D_4=-x_4D_1$, &  & \\ 
& $(g_a^{\circ}-\Lambda)D_2D_4=-x_4D_2$, &  & \\ 
\hline
\multirow{8}{*}{$C$} & $g_2D_1D_2-(g_2-\Lambda_a)D_2D_1=-x_1D_2, \quad g_3D_1D_3-(g_3-\Lambda_a)D_3D_1=-x_1D_3$, & 
\multirow{3}{*}{$I=\{1\}$, $R_a=\{ 2, 3, 4 \}$} & \multirow{3}{*}{$g_2, g_3, g_4 \ne \Lambda_a$} \\ 
& $g_4D_1D_4-(g_4-\Lambda_a)D_4D_1=-x_1D_4$, &  & \\ \cline{2-4}
& $g_1D_2D_1-(g_1-\Lambda_a)D_1D_2=-x_2D_1, \quad g_3D_2D_3-(g_3-\Lambda_a)D_3D_2=-x_2D_3$, & 
\multirow{3}{*}{$I=\{2\}$, $R_a=\{ 1, 3, 4 \}$} & \multirow{3}{*}{$g_1 \ne 0$, $g_3, g_4 \ne \Lambda_a$} \\ 
& $g_4D_2D_4-(g_4-\Lambda_a)D_4D_2=-x_2D_4$, &  & \\ \cline{2-4}
& $g_1D_3D_1-(g_1-\Lambda_a)D_1D_3=-x_3D_1, \quad g_2D_3D_2-(g_2-\Lambda_a)D_2D_3=-x_3D_2$, & 
\multirow{3}{*}{$I=\{3\}$, $R_a=\{ 1, 2, 4 \}$} & \multirow{3}{*}{$g_1, g_2 \ne 0$, $g_4 \ne \Lambda_a$} \\ 
& $g_4D_3D_4-(g_4-\Lambda_a)D_4D_3=-x_3D_4$, &  & \\ \cline{2-4}
& $g_1D_4D_1-(g_1-\Lambda_a)D_1D_4=-x_4D_1, \quad g_2D_4D_2-(g_2-\Lambda_a)D_2D_4=-x_4D_2$, & 
\multirow{3}{*}{$I=\{4\}$, $R_a=\{ 1, 2, 3 \}$} & \multirow{3}{*}{$g_1, g_2, g_3 \ne 0$} \\ 
& $g_3D_4D_3-(g_3-\Lambda_a)D_3D_4=-x_4D_3$, &  & \\ 
\hline
\multirow{3}{*}{$D$} & $D_1D_2-q_{21}D_2D_1=0, \quad D_1D_3-q_{31}D_3D_1=0 $, & \multirow{3}{*}{$R=\{1, 2, 3, 4\}$} & \multirow{3}{*}{$q_{21}, q_{31}, q_{41}, q_{32}, q_{42}, q_ {43} \ne 0$} \\
& $D_1D_4-q_{41}D_4D_1=0, \quad D_2D_3-q_{32}D_3D_2=0,$ & & \\
& $D_2D_4-q_{42}D_4D_2=0, \quad D_3D_4-q_{43}D_4D_3=0$ & & \\
\hline
\end{tabular}
}
\end{center}
\end{table}

\subsection{Diffusion algebras on five generators}\label{ClasDF5}

As in Subsection \ref{ClasDF4}, all possible diffusion algebras on five generators are shown. These can be found in Tables \ref{Diffusion5(1)}, \ref{Diffusion5(2)}, \ref{Diffusion5(3)}, \ref{Diffusion5(4)}, \ref{Diffusion5(5)}, \ref{Diffusion5(6)}, \ref{Diffusion5(7)}, \ref{Diffusion5(8)}, \ref{Diffusion5(9)} and \ref{Diffusion5(10)}.

\begin{table}[h]
\caption{Diffusion algebras on five generators}
\label{Diffusion5(1)}
\begin{center}
\resizebox{12.6cm}{!}{
\setlength\extrarowheight{6pt}
\begin{tabular}{ |c|c|c|c| } 
\hline
{\rm Diffusion algebra} & {\rm Relations} & {\rm Conditions} & {\rm Restrictions} \\
\hline
\multirow{55}{*}{$A_I$} & $gD_2D_1-gD_1D_2=x_1D_2-x_2D_1, \quad gD_3D_1-gD_1D_3=x_1D_3-x_3D_1$, & 
\multirow{5}{*}{$I=\{1, 2, 3, 4\}$, $S=\{5\}$} & \multirow{5}{*}{$g, g_5 \ne 0$} \\
& $gD_4D_1-gD_1D_4=x_1D_4-x_4D_1, \quad g_5D_5D_1-g_5D_1D_5=x_1D_5$, & & \\ 
& $gD_3D_2-gD_2D_3=x_2D_3-x_3D_2, \quad gD_4D_2-gD_2D_4=x_2D_4-x_4D_2$, & & \\ 
& $g_5D_5D_2-g_5D_2D_5=x_2D_5, \quad gD_4D_3-gD_3D_4=x_3D_4-x_4D_3$, & & \\ 
& $g_5D_5D_3-g_5D_3D_5=x_3D_5, \quad g_5D_5D_4-g_5D_4D_5=x_4D_5$ & & \\ \cline{2-4}
& $gD_2D_1-gD_1D_2=x_1D_2-x_2D_1, \quad gD_3D_1-gD_1D_3=x_1D_3-x_3D_1$, & 
\multirow{5}{*}{$I=\{1, 2, 3, 5\}$, $S=\{4\}$} & \multirow{5}{*}{$g, g_4 \ne 0$} \\
& $gD_{5}D_1-gD_1D_{5}=x_1D_{5}-x_{5}D_1, \quad g_4D_4D_1-g_4D_1D_4=x_1D_4$, & & \\
\hline
\end{tabular}
}
\end{center}
\end{table}

\begin{table}[h]
\caption{Diffusion algebras on five generators}
\label{Diffusion5(2)}
\begin{center}
\resizebox{12.6cm}{!}{
\setlength\extrarowheight{6pt}
\begin{tabular}{ |c|c|c|c| } 
\hline
{\rm Diffusion algebra} & {\rm Relations} & {\rm Conditions} & {\rm Restrictions} \\
\hline
\multirow{60}{*}{$A_I$} & $gD_2D_1-gD_1D_2=x_1D_2-x_2D_1, \quad gD_5D_1-gD_1D_5=x_1D_5-x_5D_1$, & 
\multirow{5}{*}{$I=\{1, 2, 5\}$, $S=\{3\}$, $T_b^{\bullet}=\{4\}$} & \multirow{5}{*}{$g, g_3, g_b^{\pm} \ne 0$} \\
& $gD_5D_2-gD_2D_5=x_2D_5-x_5D_2, \quad g_3D_3D_1-g_3D_1D_3=x_1D_3$, & & \\ 
& $g_3D_3D_2-g_3D_2D_3=x_2D_3, \quad g_3D_3D_5-g_3D_5D_3=x_5D_3$, & & \\ 
& $g_b^{+}D_1D_4=-x_1D_4, \quad g_b^{+}D_2D_4=-x_2D_4$, & & \\ 
& $g_b^{-}D_4D_5=x_5D_4$ & & \\ \cline{2-4}
& $gD_2D_1-gD_1D_2=x_1D_2-x_2D_1, \quad gD_5D_1-gD_1D_5=x_1D_5-x_5D_1$, & 
\multirow{5}{*}{$I=\{1, 2, 5\}$, $S=\{4\}$, $T_b^{\bullet}=\{3\}$} & \multirow{5}{*}{$g, g_4, g_b^{\pm} \ne 0$} \\
& $gD_5D_2-gD_2D_5=x_2D_5-x_5D_2, \quad g_4D_4D_1-g_4D_1D_4=x_1D_4$, & & \\
\hline
\end{tabular}
}
\end{center}
\end{table}

\begin{table}[h]
\caption{Diffusion algebras on five generators}
\label{Diffusion5(3)}
\begin{center}
\resizebox{12.6cm}{!}{
\setlength\extrarowheight{6pt}
\begin{tabular}{ |c|c|c|c| } 
\hline
{\rm Diffusion algebra} & {\rm Relations} & {\rm Conditions} & {\rm Restrictions} \\
\hline
\multirow{60}{*}{$A_{I}$} & $gD_2D_4-gD_4D_2=x_4D_2-x_2D_4, \quad gD_3D_4-gD_4D_3=x_4D_3-x_3D_4$, & 
\multirow{5}{*}{$I=\{2, 3, 4\}$, $S=\{1\}$, $T_a^{\circ}=\{5\}$} & \multirow{5}{*}{$g, g_1, g_a^{\circ} \ne 0$} \\
& $gD_3D_2-gD_2D_3=x_2D_3-x_3D_2, \quad g_1D_1D_4-g_1D_4D_1=x_4D_1$, & & \\ 
& $g_1D_1D_2-g_1D_2D_1=x_2D_1, \quad g_1D_1D_3-g_1D_3D_1=x_3D_1$, & & \\ 
& $g_a^{\circ}D_4D_5=-x_4D_5, \quad g_a^{\circ}D_2D_5=-x_2D_5$, & & \\ 
& $g_a^{\circ}D_3D_5=-x_3D_5$ & & \\ \cline{2-4}
& $gD_2D_4-gD_4D_2=x_4D_2-x_2D_4, \quad gD_3D_4-gD_4D_3=x_4D_3-x_3D_4$, & 
\multirow{5}{*}{$I=\{2, 3, 4\}$, $S=\{5\}$, $T_a^{\circ}=\{1\}$} & \multirow{5}{*}{$g, g_5, g_a^{\circ} \ne 0$} \\
& $gD_3D_2-gD_2D_3=x_2D_3-x_3D_2, \quad g_5D_5D_4-g_5D_4D_5=x_4D_5$, & & \\
\hline
\end{tabular}
}
\end{center}
\end{table}

\begin{table}[h]
\caption{Diffusion algebras on five generators}
\label{Diffusion5(4)}
\begin{center}
\resizebox{12.6cm}{!}{
\setlength\extrarowheight{6pt}
\begin{tabular}{ |c|c|c|c| } 
\hline
{\rm Diffusion algebra} & {\rm Relations} & {\rm Conditions} & {\rm Restrictions} \\
\hline
\multirow{55}{*}{$A_{II}$} & $(g_1-g_2)D_1D_2= x_2D_1-x_1D_2, \quad (g_1-g_3)D_1D_3= x_3D_1-x_1D_3$, & 
\multirow{5}{*}{$I=\{1, 2, 3, 4, 5\}$} & \multirow{5}{*}{$g_i \ne g_j$ for $i, j \in I$, $i \ne j$} \\ 
& $(g_1-g_4)D_1D_4= x_4D_1-x_1D_4, \quad (g_1-g_5)D_1D_5= x_5D_1-x_1D_5$, & & \\
& $(g_2-g_3)D_2D_3= x_3D_2-x_2D_3, \quad (g_2-g_4)D_2D_4= x_4D_2-x_2D_4$, & & \\
& $(g_2-g_5)D_2D_5= x_5D_2-x_2D_5, \quad (g_3-g_4)D_3D_4= x_4D_3-x_3D_4$, & & \\
& $(g_3-g_5)D_3D_5= x_5D_3-x_3D_5, \quad (g_4-g_5)D_4D_5= x_5D_4-x_4D_5$ & & \\ \cline{2-4}
& $(g_1-g_2)D_1D_2= x_2D_1-x_1D_2, \quad (g_1-g_3)D_1D_3= x_3D_1-x_1D_3$, & 
\multirow{5}{*}{$I=\{1, 2, 3, 4\}$, $T_a^{\circ}=\{5\}$} & \multirow{5}{*}{$g_i \ne g_j$, $i \ne j$, $g_i \ne g_a^{\circ}$, for $i, j \in I$} \\ 
& $(g_1-g_4)D_1D_4= x_4D_1-x_1D_4, \quad (g_2-g_3)D_2D_3= x_3D_2-x_2D_3$, & & \\
\hline
\end{tabular}
}
\end{center}
\end{table}

\begin{table}[h]
\caption{Diffusion algebras on five generators}
\label{Diffusion5(5)}
\begin{center}
\resizebox{12.6cm}{!}{
\setlength\extrarowheight{6pt}
\begin{tabular}{ |c|c|c|c| } 
\hline
{\rm Diffusion algebra} & {\rm Relations} & {\rm Conditions} & {\rm Restrictions} \\
\hline
\multirow{20}{*}{$A_{II}$}& $(g_2-g_4)D_2D_4= x_4D_2-x_2D_4, \quad (g_2-g_5)D_2D_5= x_5D_2-x_2D_5$, & 
\multirow{5}{*}{$I=\{2, 4, 5\}$, $T_b^{\bullet}=\{3\}$ $T_a^{\circ}=\{1\}$} & \multirow{5}{*}{$g_i \ne g_j$, $i \ne j$, $g_i \not \in \{g_a^{\circ}, \mp g_b^{\pm}\}$, for $i, j \in I$} \\ 
& $(g_4-g_5)D_4D_5= x_5D_4-x_4D_5, \quad (g_2+g_b^{+})D_2D_3= -x_2D_3$, & & \\
& $(g_b^{-}-g_4)D_3D_4= x_4D_3, \quad (g_b^{-}-g_5)D_3D_5= x_5D_3$, & & \\
& $(g_2+g_a^{\circ})D_1D_2= -x_2D_1, \quad (g_4+g_a^{\circ})D_1D_4= -x_4D_1$, & & \\
& $(g_5+g_a^{\circ})D_1D_5= -x_1D_5$ & & \\ \cline{2-4}
& $(g_1-g_2)D_1D_2= x_2D_1-x_1D_2, \quad (g_1-g_5)D_1D_5= x_5D_1-x_1D_5$, & 
\multirow{5}{*}{$I=\{1, 2, 5\}$, $T_b^{\bullet}=\{3, 4\}$} & \multirow{5}{*}{$g_i \ne g_j$, $i \ne j$, $g_i \not \in \{ \mp g_b^{\pm}\}$, for $i, j \in I$} \\ 
& $(g_2-g_5)D_2D_5= x_5D_2-x_2D_5, \quad (g_1+g_b^{+})D_1D_3= -x_1D_3$, & & \\
\hline
\end{tabular}
}
\end{center}
\end{table}

\begin{table}[h]
\caption{Diffusion algebras on five generators}
\label{Diffusion5(6)}
\begin{center}
\resizebox{12.6cm}{!}{
\setlength\extrarowheight{6pt}
\begin{tabular}{ |c|c|c|c| } 
\hline
{\rm Diffusion algebra} & {\rm Relations} & {\rm Conditions} & {\rm Restrictions} \\
\hline
\multirow{68}{*}{$B$} & $gD_1D_3-(g-\Lambda)D_3D_1= x_3D_1-x_1D_3, \quad g_4D_1D_4-(g_4-\Lambda)D_4D_1= -x_1D_4$, & 
\multirow{4}{*}{$I=\{1, 3\}$, $S=\{4\}$, $T_b^{\bullet}=\{2\}$, $T_a^{\circ}=\{5\}$} & \multirow{4}{*}{$g, g_b^{\pm} \ne 0$, $g_4, g_a^{\circ} \not \in \{0, \Lambda \}$} \\ 
& $g_4D_4D_3-(g_4-\Lambda)D_3D_4= x_3D_4, \quad g_b^{+}D_1D_2=-x_1D_2$, & & \\
& $g_b^{-}D_2D_3=x_3D_2, \quad g_a^{\circ}D_1D_5=-x_1D_5$, & & \\
& $(g_a^{\circ}-\Lambda)D_3D_5= -x_3D_5$, & & \\ \cline{2-4}
& $gD_1D_3-(g-\Lambda)D_3D_1= x_3D_1-x_1D_3, \quad g_5D_1D_5-(g_5-\Lambda)D_5D_1= -x_1D_5$, & 
\multirow{4}{*}{$I=\{1, 3\}$, $S=\{5\}$, $T_b^{\bullet}=\{2\}$, $T_a^{\circ}=\{4\}$} & \multirow{4}{*}{$g, g_b^{\pm} \ne 0$, $g_5, g_a^{\circ} \not \in \{0, \Lambda \}$} \\ 
& $g_5D_5D_3-(g_5-\Lambda)D_3D_5= x_3D_5, \quad g_b^{+}D_1D_2=-x_1D_2$, & & \\
& $g_b^{-}D_2D_3=x_3D_2, \quad g_a^{\circ}D_1D_4=-x_1D_4$, & & \\
\hline
\end{tabular}
}
\end{center}
\end{table}

\begin{table}[h]
\caption{Diffusion algebras on five generators}
\label{Diffusion5(7)}
\begin{center}
\resizebox{12.6cm}{!}{
\setlength\extrarowheight{6pt}
\begin{tabular}{ |c|c|c|c| } 
\hline
{\rm Diffusion algebra} & {\rm Relations} & {\rm Conditions} & {\rm Restrictions} \\
\hline
\multirow{72}{*}{$B$}& $gD_1D_5-(g-\Lambda)D_5D_1= x_5D_1-x_1D_5, \quad g_4D_1D_4-(g_4-\Lambda)D_4D_1= -x_1D_4$, & 
\multirow{4}{*}{$I=\{1, 5\}$, $S=\{4\}$, $T_b^{\bullet}=\{2, 3\}$} & \multirow{4}{*}{$g, g_4, g_b^{\pm} \ne 0$} \\ 
& $g_4D_4D_5-(g_4-\Lambda)D_5D_4= x_5D_4, \quad g_b^{+}D_1D_2=-x_1D_2$, & & \\
& $g_b^{-}D_2D_5=x_5D_2, \quad g_b^{+}D_1D_3=-x_1D_3$, & & \\
& $g_b^{-}D_3D_5= x_5D_3$, & & \\ \cline{2-4}
& $gD_1D_5-(g-\Lambda)D_5D_1= x_5D_1-x_1D_5, \quad g_b^{+}D_1D_4=-x_1D_4$, & 
\multirow{4}{*}{$I=\{1, 5\}$, $T_b^{\bullet}=\{2,3,4 \}$} & \multirow{4}{*}{$g, g_b^{\pm}\ne 0$} \\ 
& $g_b^{-}D_4D_5= x_5D_4, \quad g_b^{+}D_1D_2=-x_1D_2$, & & \\
& $g_b^{-}D_2D_5=x_5D_2, \quad g_b^{+}D_1D_3=-x_1D_3$, & & \\
\hline
\end{tabular}
}
\end{center}
\end{table}

\begin{table}[h]
\caption{Diffusion algebras on five generators}
\label{Diffusion5(8)}
\begin{center}
\resizebox{12.6cm}{!}{
\setlength\extrarowheight{6pt}
\begin{tabular}{ |c|c|c|c| } 
\hline
{\rm Diffusion algebra} & {\rm Relations} & {\rm Conditions} & {\rm Restrictions} \\
\hline
\multirow{72}{*}{$B$} & $gD_2D_5-(g-\Lambda)D_5D_2= x_5D_2-x_2D_5, \quad g_1D_2D_1-(g_1-\Lambda)D_1D_2= -x_2D_1$, & 
\multirow{4}{*}{$I=\{2, 5\}$, $S=\{1, 3, 4\}$} & \multirow{4}{*}{$g, g_3, g_4\ne 0$, $g_1 \not \in \{0, \Lambda \}$} \\ 
& $g_1D_1D_5-(g_1-\Lambda)D_5D_1= x_5D_1, \quad g_3D_2D_3-(g_3-\Lambda)D_3D_2= -x_2D_3$, & & \\
& $g_3D_3D_5-(g_3-\Lambda)D_5D_3= x_5D_3, \quad g_4D_2D_4-(g_4-\Lambda)D_4D_2= -x_2D_4$, & & \\
& $g_4D_4D_5-(g_4-\Lambda)D_5D_4= x_5D_4$, & & \\ \cline{2-4}
& $gD_2D_5-(g-\Lambda)D_5D_2= x_5D_2-x_2D_5, \quad g_1D_2D_1-(g_1-\Lambda)D_1D_2= -x_2D_1$, & 
\multirow{4}{*}{$I=\{2, 5\}$, $S=\{1, 3\}$, $T_b^{\bullet}=\{4\}$} & \multirow{4}{*}{$g, g_3, g_b^{\pm}\ne 0$, $g_1 \not \in \{0, \Lambda \}$ } \\ 
& $g_1D_1D_5-(g_1-\Lambda)D_5D_1= x_5D_1, \quad g_3D_2D_3-(g_3-\Lambda)D_3D_2= -x_2D_3$, & & \\
& $g_3D_3D_5-(g_3-\Lambda)D_5D_3= x_5D_3, \quad g_b^{+}D_2D_4=-x_2D_4$, & & \\
\hline
\end{tabular}
}
\end{center}
\end{table}

\begin{table}[h]
\caption{Diffusion algebras on five generators}
\label{Diffusion5(9)}
\begin{center}
\resizebox{12.6cm}{!}{
\setlength\extrarowheight{6pt}
\begin{tabular}{ |c|c|c|c| } 
\hline
{\rm Diffusion algebra} & {\rm Relations} & {\rm Conditions} & {\rm Restrictions} \\
\hline
\multirow{72}{*}{$B$} & $gD_5D_3-(g-\Lambda)D_3D_5= x_3D_5-x_5D_3, \quad g_2D_5D_2-(g_2-\Lambda)D_2D_5= -x_5D_2$, & 
\multirow{4}{*}{$I=\{3, 5\}$, $S=\{1, 2\}$, $T_b^{\bullet}=\{4\}$} & \multirow{4}{*}{$g\ne 0$, $g_s, g_a^{\circ} \not \in \{0, \Lambda \}$, $s\in S$ } \\ 
& $g_2D_2D_3-(g_2-\Lambda)D_3D_2= x_3D_2, \quad g_1D_5D_1-(g_1-\Lambda)D_1D_5= -x_5D_1$, & & \\
& $g_1D_1D_3-(g_1-\Lambda)D_3D_1= x_3D_1, \quad g_b^{+}D_3D_4=-x_3D_4$, & & \\
& $g_b^{-}D_4D_5= x_5D_4$, & & \\ \cline{2-4}
& $gD_5D_3-(g-\Lambda)D_3D_5= x_3D_5-x_5D_3, \quad g_4D_5D_4-(g_4-\Lambda)D_4D_5= -x_5D_4$, & 
\multirow{4}{*}{$I=\{3, 5\}$, $S=\{1, 4\}$, $T_a^{\circ}=\{2\}$} & \multirow{4}{*}{$g, g_4\ne 0$, $g_1, g_a^{\circ} \not \in \{0, \Lambda \}$} \\ 
& $g_4D_4D_3-(g_4-\Lambda)D_3D_4= x_3D_4, \quad g_1D_5D_1-(g_1-\Lambda)D_1D_5= -x_5D_1$, & & \\
& $g_1D_1D_3-(g_1-\Lambda)D_3D_1= x_3D_1, \quad g_a^{\circ}D_2D_5=-x_5D_2$, & & \\
\hline
\end{tabular}
}
\end{center}
\end{table}

\begin{table}[h]
\caption{Diffusion algebras on five generators}
\label{Diffusion5(10)}
\begin{center}
\resizebox{12.6cm}{!}{
\setlength\extrarowheight{6pt}
\begin{tabular}{ |c|c|c|c| } 
\hline
{\rm Diffusion algebra} & {\rm Relations} & {\rm Conditions} & {\rm Restrictions} \\
\hline

\hline
\multirow{10}{*}{C} & $g_2D_1D_2-(g_2-\Lambda_a)D_2D_1=-x_1D_2, \quad g_3D_1D_3-(g_3-\Lambda_a)D_3D_1=-x_1D_3$, & 
\multirow{3}{*}{$I=\{1\}$, $R_a=\{ 2, 3, 4, 5 \}$} & \multirow{3}{*}{$g_2, g_3, g_4, g_5 \ne \Lambda_a$} \\ 
& $g_4D_1D_4-(g_4-\Lambda_a)D_4D_1=-x_1D_4, \quad g_5D_1D_5-(g_5-\Lambda_a)D_5D_1=-x_1D_5$, &  & \\ \cline{2-4}
& $g_1D_2D_1-(g_1-\Lambda_a)D_1D_2=-x_2D_1, \quad g_3D_2D_3-(g_3-\Lambda_a)D_3D_2=-x_2D_3$, & 
\multirow{3}{*}{$I=\{2\}$, $R_a=\{ 1, 3, 4, 5 \}$} & \multirow{3}{*}{$g_1 \ne 0$, $g_3, g_4, g_5 \ne \Lambda_a$} \\ 
& $g_4D_2D_4-(g_4-\Lambda_a)D_4D_2=-x_2D_4, \quad g_5D_2D_5-(g_5-\Lambda_a)D_5D_2=-x_2D_5$,  &  & \\ \cline{2-4}
& $g_1D_3D_1-(g_1-\Lambda_a)D_1D_3=-x_3D_1, \quad g_2D_3D_2-(g_2-\Lambda_a)D_2D_3=-x_3D_2$, & 
\multirow{3}{*}{$I=\{3\}$, $R_a=\{ 1, 2, 4, 5 \}$} & \multirow{3}{*}{$g_1, g_2 \ne 0$, $g_4, g_5 \ne \Lambda_a$} \\ 
& $g_4D_3D_4-(g_4-\Lambda_a)D_4D_3=-x_3D_4, \quad g_5D_3D_5-(g_5-\Lambda_a)D_5D_3=-x_3D_5$, &  & \\ \cline{2-4}
& $g_1D_4D_1-(g_1-\Lambda_a)D_1D_4=-x_4D_1, \quad g_2D_4D_2-(g_2-\Lambda_a)D_2D_4=-x_4D_2$, & 
\multirow{3}{*}{$I=\{4\}$, $R_a=\{ 1, 2, 3, 5 \}$} & \multirow{3}{*}{$g_1, g_2, g_3 \ne 0$, $g_5 \ne \Lambda$} \\
\hline
\end{tabular}
}
\end{center}
\end{table}

\subsection{Differential smoothness of diffusion algebras on  \texorpdfstring{$N$}{Lg} generators}

If $\mathcal{D}$ is a diffusion algebra on $N$ generators, then $\text{GKdim}(\mathcal{D})=N$ \cite[Theorems 4.14 and 4.18]{Reyes2013}. Our aim in this section is to find an integral calculus of degree $N$ that guarantees the differentiable smoothness of $\mathcal{D}$.

\begin{theorem}\label{DSdiffalgN}
Let $\mathcal{D}$ be a diffusion algebra on $N\geq 3$ generators. If $\mathcal{D}$ satisfies any one of the following properties:
\begin{enumerate}
    \item [\rm (i)] $|I|=L$ and $|S|=N-L$, for $3\leq L \leq N$; 
    \item [\rm (ii)] $|I|=1$, $|S|= N-1$ and $g_s=\mathcal{G}\not\in \{0, \Lambda_a\}$ for all $s\in S$; 
    \item [\rm (iii)] $|I|=2$, $|S|=N-2$ and $g_s=\mathcal{G}\not\in \{0, \Lambda\}$ for all $s\in S$; or
    \item [\rm (iii)] $|I|=0$,
\end{enumerate}
\end{theorem}

then $\mathcal{D}$ is diferentially smooth.
\begin{proof}
From Section \ref{DefinitionsandpreliminariesDSA} we know that we must to consider $\Omega^{1}(\mathcal{D})$, a free right $\mathcal{D}$-module of rank $N$ with generators $dD_i$, $1\leq i \leq N$. Define a left $\mathcal{D}$-module structure by
    \begin{equation}\label{relrightmoddiffN}
        pdD_a = dD_a \nu_{D_a}(p), \quad {\rm for}\ a \in I \cup S, \ p\in \mathcal{D},
    \end{equation}
    
where $\nu_{D_a}$ are algebra automorphisms of $\mathcal{D}$.
\begin{itemize}
    \item [\rm (i)] Consider the maps given by
\begin{align}
   \nu_{D_i}(D_j) = &\ D_j-g^{-1}x_j, \text{ for } i,j \in I, \label{auto1case1N} \\ 
    \nu_{D_i}(D_s) = &\ D_s, \text{ for } i \in I, \ s\in S \label{auto2case1N} \\ 
    \nu_{D_s}(D_i) = &\ D_i-g_{s}^{-1}x_i, \text{ for } i \in I, \ s\in S, \quad {\rm and} \label{auto3case1N} \\ 
    \nu_{D_s}(D_l) = &\ D_l, \text{ for } s,l \in S. \label{auto4case1N}
\end{align}

It can be seen that the maps $\nu_{D_i}$ and $\nu_{D_s}$ for each $i \in I$ and $s \in S$ can be extended to algebra homomorphisms of $\mathcal{D}$. As a matter of fact, these maps respect the relations  {\rm (}\ref{PyatovTwarock2002(31)}{\rm )} for $\mathcal{D}$ and commute with each other.

Consider $\Omega^{1}(\mathcal{D})$ a free right $\mathcal{D}$-module of rank $N$ with generators $dD_i$, $dD_s$ for all $i\in I$ and $s\in S$. For every element $p\in \mathcal{D}$ define a left $\mathcal{D}$-module structure by
\begin{align}
    pdD_a = dD_a \nu_{D_a}(p), \quad {\rm with}\ a\in I \cup S \label{relrightmoddiff1N}.
\end{align}

The relations in $\Omega^{1}(\mathcal{D})$ are given by 
\begin{align}
D_idD_j = &\ dD_j (D_i-g^{-1}x_i), \quad i, j \in I, \notag \\
D_idD_s = &\ dD_s(D_i-g_s^{-1}x_i), \quad i\in I, s\in S, \notag \\
D_sdD_i = &\ dD_iD_s, \quad i\in I, s\in S, \ {\rm and} \label{relDiffD1N}\\
D_sdD_l = &\ dD_lD_s, \quad s, l \in S. \notag  
\end{align}

We want to extend the correspondences 
\begin{equation*}
D_a \mapsto d D_a, \quad {\rm for\ every}\ a\in I \cup S, 
\end{equation*} 

to a map $d: \mathcal{D} \to \Omega^{1}(\mathcal{D})$ satisfying the Leibniz's rule. This is possible if it is compatible with the nontrivial relations {\rm (}\ref{PyatovTwarock2002(31)}{\rm )}, i.e. if the following two equalities
\begin{align*}
        gdD_iD_j+gD_idD_j-gdD_jD_i-gD_jdD_i &\ = x_jdD_i-x_idD_j, \quad \text{ for } i,j \in I, \ {\rm and} \\
        g_sdD_sD_i+g_sD_sdD_i-g_sdD_iD_s-g_sD_idD_s &\ = x_idD_s, \quad \text{ for } i\in I, \ s \in S, 
\end{align*}

hold.

Define $\Bbbk$-linear maps 
\begin{equation*}
\partial_{D_a}: \mathcal{D} \rightarrow \mathcal{D}, \quad {\rm for\ every}\ a\in I \cup S,
\end{equation*}

such that
\begin{align*}
    d(p)=\sum_{a\in I\cup S}dD_a\partial_{D_a}(p), \quad {\rm for\ all} \ p \in \mathcal{D}.
\end{align*}

Since $dD_a$ with $a \in I \cup S$ are free generators of the right $\mathcal{D}$-module $\Omega^1(\mathcal{D})$, these maps are well-defined. Note that $d(p)=0$ if and only if $\partial_{D_a}(p)=0$ for each $a\in I \cup S$. By using the relations in {\rm (}\ref{relDiffD1N}{\rm )} and the definitions of the maps $\nu_{D_a}$ with $a\in I \cup S$, we get that
\begin{align*}
\partial_{D_a}(D_1^{k_1}\cdots D_N^{k_N}) = &\ {k_a}(D_1-g_a^{-1}x_1)^{k_1}\cdots (D_{a-1}-g_{a}^{-1}x_{a-1})^{k_{a-1}}D_a^{k_a-1}\cdots D_N^{k_N},
\end{align*}

where $g_{a}^{-1}=g^{-1}$ if $a\in I$, and $g_{a}^{-1}=g_{s}^{-1}$ if $r=s\in S$. 
Thus, $d(p)=0$ if and only if $p$ is a scalar multiple of the identity. This shows that $(\Omega \mathcal{D},d)$ is connected with $\Omega (\mathcal{D}) = \bigoplus_{k=0}^{n-1}\Omega^k (\mathcal{D})$.

The universal extension of $d$ to higher forms compatible with {\rm (}\ref{relrightmoddiff1N}{\rm )} gives the following rules for $\Omega^k(\mathcal{D})$ ($2\leq k \leq N-1$):
\begin{align}\label{relnwedgediff1N}
    \bigwedge_{r=1}^{k}dD_{q(r)} =  (-1)^{\sharp}\bigwedge_{r=1}^kdD_{p(r)}, 
\end{align}

where $q:\{1,\ldots,k\}\rightarrow \{1,\ldots,N\}$ is an injective map, $p:\{1,\ldots,k\}\rightarrow \text{Im}(q)$ is an increasing injective map and $\sharp$ is the number of $2$-permutation needed to transform $q$ into $p$.

By using that the automorphisms $\nu_{D_i}$ and $\nu_{D_s}$, for every $i \in I$ and $s\in S$, commute with each other, there are no additional relationships to the previous ones. In this way, 
\begin{align*}
 \Omega^{N-1}(\mathcal{D}) = &\ \left[dD_1\wedge dD_2\wedge \cdots \wedge dD_{N-1}\oplus dD_1\wedge dD_3\wedge \cdots \wedge dD_{N} \right. \\
    &\ \left. \oplus \cdots \oplus \ dD_2\wedge \cdots \wedge dD_{N}\right]\mathcal{D}. 
\end{align*}

Since $\Omega^N(\mathcal{D}) = \omega\mathcal{D}\cong \mathcal{D}$ as a right and left $\mathcal{D}$-module with 
$$
\omega = dD_1\wedge \cdots \wedge dD_N \quad {\rm and} \quad \nu_{\omega}=\nu_{D_1}\circ\cdots\circ\nu_{D_N}
$$

we have that $\omega$ is a volume form of $\mathcal{D}$. From Proposition \ref{BrzezinskiSitarz2017Lemmas2.6and2.7} (2) we obtain that $\omega$ is an integral form by setting
\begin{align*}
\omega_i^j = &\ \bigwedge_{k=1}^{j}dD_{p_{i,j}(k)}, \quad \text{ for } 1\leq i \leq \binom{N}{j}, \quad {\rm and} \\
\bar{\omega}_i^{N-j} = &\ (-1)^{\sharp_i}\bigwedge_{k=j+1}^{n}dD_{\bar{p}_{i,j}(k)},  \quad \text{for}\ 1\leq i \leq \binom{N}{j},
\end{align*}

where
$$
p_{i,j}:\{1,\ldots,j\}\rightarrow \{1,\ldots,N\} \quad {\rm and} \quad \bar{p}_{i,j}:\{j+1,\ldots,N\}\rightarrow (\text{Im}(p_{i,j}))^c
$$

are increasing injective maps and $\sharp_{i,j}$ is the number of $2$-permutation needed to transform $\{\bar{p}_{i,j}(j+1),\ldots, \bar{p}_{i,j}(N),p_{i,j}(1), \ldots, p_{i,j}(j)\}$ into $\{1, \ldots, N\}$. 

Let $\omega' \in \Omega^j(\mathcal{D})$. Then:
\[
\omega' = \displaystyle\sum_{i=1}^{\binom{N}{j}}\bigwedge_{k=1}^{j}dD_{p_{i,j}(k)}a_i, \quad {\rm with}\ a_i \in \Bbbk.
\] 

This implies that we have the equalities given by
{\footnotesize{
    \begin{align*}
    \sum_{i=1}^{\binom{N}{j}}\omega_{i}^{j}\pi_{\omega}(\bar{\omega}_i^{N-j}\wedge \omega') &\ =\sum_{i=1}^{\binom{N}{j}}\bigwedge_{k=1}^{j}dD_{p_{i,j}(k)}\pi_{\omega}\left(a_i(-1)^{\sharp_{i,j}}\bigwedge_{k=j+1}^{N}dD_{\bar{p}_{i,j}(k)}\wedge\bigwedge_{k=1}^{j}dD_{p_{i,j}(k)}\right) \\
    &\ = \sum_{i=1}^{\binom{N}{j}}\bigwedge_{k=1}^{j}dx_{p_{i,j}(k)}a_i \\ 
    &\ = \omega'.
    \end{align*}
}}

By Proposition \ref{BrzezinskiSitarz2017Lemmas2.6and2.7} (2) it follows that $\mathcal{D}$ is differentially smooth.

\item [\rm (ii)] Consider the maps given by
\begin{align}
    \nu_{D_{\bf i}}(D_{\bf i}) = &\ D_{\bf i}, \notag \\
    \nu_{D_{\bf i}}(D_s) = &\ \mathcal{G}(\mathcal{G}-\Lambda_a)^{-1}D_s, \label{auto1case2N} \\
    \nu_{D_s}(D_{\bf i}) = &\ \mathcal{G}^{-1}((\mathcal{G}-\Lambda_a)D_{\bf i}-x_{\bf i}),  \notag \\
    \nu_{D_s}(D_{l}) = &\ D_{l}, \quad \text{ for } \ s, l \in S.  \label{auto2case2N}
\end{align}

It is straightforward to show that the maps $\nu_{D_{\bf i}}$ and $\nu_{D_s}$ for each $s \in S$ can be extended to algebra homomorphisms of $\mathcal{D}$ and respect the relations  {\rm (}\ref{PyatovTwarock2002(34)}{\rm )} for $\mathcal{D}$. Again, these maps commute with each other.

Consider $\Omega^{1}(\mathcal{D})$ a free right $\mathcal{D}$-module of rank $N$ with generators $dD_{\bf i}$ for every $dD_s$ and $s\in S$. For all $p\in \mathcal{D}$ define a left $\mathcal{D}$-module structure by
\begin{align}
    pdD_a = dD_a \nu_{D_a}(p), \quad {\rm for \ every}\ a\in \{\textbf{i}\} \cup S \label{relrightmoddiff2N}.
\end{align}

The relations in $\Omega^{1}(\mathcal{D})$ are given by 
\begin{align}
D_{\bf i} dD_s = &\ dD_s\mathcal{G}^{-1}((\mathcal{G}-\Lambda_a)D_{\bf i}-x_{\bf i}), \quad s\in S, \notag \\
D_sdD_{\bf i} = &\  dD_{\bf i}\mathcal{G}(\mathcal{G}-\Lambda_a)^{-1}D_s, \quad s\in S, \ {\rm and} \notag \\
 D_sdD_l = &\ dD_lD_s, \quad  s, l \in S. \label{relDiffD2N}
\end{align}

We want to extend the correspondences 
\begin{equation*}
D_a \mapsto d D_a, \quad {\rm with} \ a\in \{\textbf{i}\} \cup S
\end{equation*} 

to a map $d: \mathcal{D} \to \Omega^{1}(\mathcal{D})$ satisfying the Leibniz's rule. By using the relations {\rm (}\ref{PyatovTwarock2002(34)}{\rm )} the equality
\begin{align*}
        \mathcal{G} dD_{\bf i} D_s+\mathcal{G} D_{\bf i} dD_s - (\mathcal{G} - \Lambda_a) dD_s D_{\bf i}-(\mathcal{G} - \Lambda_a) D_s dD_{\bf i} = -x_{\bf i}dD_s, \quad {\rm for\ all} \ s\in S, 
\end{align*}

must be satisfied.

Define $\Bbbk$-linear maps 
\begin{equation*}
\partial_{D_a}: \mathcal{D} \rightarrow \mathcal{D}, \quad {\rm for\ each} \ a\in \{\bf i\} \cup S,
\end{equation*}

such that
\begin{align*}
    d(p)=\sum_{a\in \{\textbf{i}\} \cup S}dD_a\partial_{D_a}(p), \quad {\rm for\ all} \ p \in \mathcal{D}.
\end{align*}

Since $dD_a$, $a \in \{\textbf{i}\} \cup S$,  are free generators of the right $\mathcal{D}$-module $\Omega^1(\mathcal{D})$, these maps are well-defined, and $d(p)=0$ if and only if $\partial_{D_a}(p)=0$, $a\in \{\textbf{i}\} \cup S$. By using the relations in {\rm (}\ref{relDiffD2N}{\rm )} and the definitions of the maps $\nu_{D_a}$, $a\in \{\textbf{i}\} \cup S$, we get that
{\footnotesize{
\begin{align*}
\partial_{D_{\bf i}}(D_1^{k_1}\cdots D_{\bf i}^{k_{\bf i}}\cdots D_N^{k_N}) = &\ \prod_{j=1}^{{\bf i}-1}\mathcal{G}^{k_j}(\mathcal{G}-\Lambda_a)^{-k_j}{k_{\bf i}}D_1^{k_1}\cdots D_{\bf i}^{k_{\bf i}-1}\cdots D_N^{k_N}, \\
\partial_{D_{s}}(D_1^{k_1}\cdots D_{\bf i}^{k_{\bf i}}\cdots D_N^{k_N}) = &\ k_{s}D_1^{k_1}\cdots D_{s}^{k_{s}-1}\cdots D_N^{k_N}, \text{ for } s< {\bf i}, \\
\partial_{D_{s}}(D_1^{k_1}\cdots D_{\bf i}^{k_{\bf i}}\cdots D_N^{k_N})  = &\ k_{s}\mathcal{G}^{-k_{\bf i}}D_1^{k_1}\cdots ((\mathcal{G}-\Lambda_a)D_{\bf i}-x_{\bf i})^{k_{\bf i}} \cdots D_{s}^{k_{s}-1}\cdots D_N^{k_N}, \text{ for } s> {\bf i}, 
\end{align*} 
}}

Thus $d(p)=0$ if and only if $p$ is a scalar multiple of the identity, whence shows $(\Omega \mathcal{D},d)$ is connected where $\Omega (\mathcal{D}) = \bigoplus_{k=0}^{n-1}\Omega^k (\mathcal{D})$.

From this treatment, the rest of the proof is completely analogous to case (i). Thus, $\mathcal{D}$ is differentially smooth.

\item [\rm (iii)] For all elements $s, l \in S$, consider the following maps
\begin{align*}
    \nu_{D_{\bf i}}(D_{\bf i}) = &\ g^{-1}((g-\Lambda)D_{\bf i}-x_{\bf i}), \\
    \quad \nu_{D_{\bf i}}(D_{\bf j}) = &\ (g-\Lambda)^{-1}(gD_{\bf j}-x_{\bf j}), \\
     \nu_{D_{\bf i}}(D_{s}) = &\ (\mathcal{G}-\Lambda)^{-1}\mathcal{G}D_s, \\
      \nu_{D_{\bf j}}(D_{\bf i}) = &\ g^{-1}((g-\Lambda)D_{\bf i}-x_{\bf i}), \\
      \nu_{D_{\bf j}}(D_{\bf j}) = &\ g^{-1}((g-\Lambda)D_{\bf j}-x_{\bf j}), \\
      \nu_{D_{\bf j}}(D_{s}) = &\ (\mathcal{G}-\Lambda)^{-1}\mathcal{G}D_s, \\
      \nu_{D_{s}}(D_{\bf i}) = &\ \mathcal{G}^{-1}((\mathcal{G}-\Lambda)D_{\bf i}-x_{\bf i}), \\
      \nu_{D_{s}}(D_{\bf j}) = &\ (\mathcal{G}-\Lambda)^{-1}(\mathcal{G}D_{\bf j}-x_{\bf j}), \quad {\rm and} \\
       \nu_{D_{s}}(D_{l}) &\ = D_{l}.
\end{align*}

Notice that these morphisms are similar to those corresponding in case (ii), so by using a similar reasoning we can prove that all of them guarantee that $\mathcal{D}$ is differentially smooth.

\item [\rm (iv)] In this case, the algebra $\mathcal{D}$ is precisely the {\em quantum affine} $N$-{\em space}, whence its differential smoothness follows from  \cite[Corollary 6 and Theorem 9]{KaracuhaLomp2014} or \cite[Corollary 4.9]{BrzezinskiLomp2018}.
\end{itemize}
\end{proof}

Following an argument similar to the one presented by Brzezi\'nki and Sitarz \cite[Example 2.5]{BrzezinskiSitarz2017}, we obtain the following.

\begin{theorem}\label{noDSDiffN}
If $|T| \ne 0$, then $\mathcal{D}$ is not differentially smooth.
\end{theorem}
\begin{proof}
Consider the sets $T^{\circ}$ and $T^{\bullet}$ as the only connective components in the decomposition in expression (\ref{Taeq}) of the subset $T \in I_N$.
\begin{enumerate}
    \item [\rm (i)] Suppose that $|T^{\circ}|\geq 1$. We can write the relations relevant to $T^{\circ}$ in the unique form given by 
\begin{align}\label{relnoDS}
    \mathcal{G} :D_iD_t: = -x_iD_t, \quad i\in I, \ t\in T_a^{\circ},
\end{align}

where $\mathcal{G}\in \{g_a^{\circ}, g_i+g_a^{\circ}, g_a^{\circ}- \Lambda \}$ depending on the case of Proposition (\ref{PyatovTwarock2002Theorem3.5}).

By applying the differential $d$ to the expression (\ref{relnoDS}) we get that
\begin{align*}
   \mathcal{G} d(:D_iD_t:) = d(-x_iD_t), \quad i\in I, \ t\in T_a^{\circ}.
\end{align*}

If $:D_iD_t:=D_iD_t$, by using that $d$ is $\Bbbk$-linear the Leibniz's rule yields that 
\begin{align*}
   \mathcal{G} dD_iD_t+ \mathcal{G}D_idD_t+x_idD_t=0, \quad i\in I, t\in T_a^{\circ}.
\end{align*}

By (\ref{BrzezinskiSitarz2017(2.2)}) the action of the module is written using automorphism $\nu_{D_t}$
\begin{align*}
    \mathcal{G} dD_iD_t+ \mathcal{G}dD_t\nu_{D_t}(D_i)+x_idD_t=0, \quad i\in I, t\in T_a^{\circ}.
\end{align*}

In this case, we have that 
\begin{align*}
    \mathcal{G}D_t=0,  \quad  t\in T_a^{\circ},
\end{align*}

which occurs only if $\mathcal{G}=0$. This fact contradicts any of the cases in Proposition (\ref{PyatovTwarock2002Theorem3.5}).

On the other hand, if $:D_iD_t:=D_tD_i$, using an argument similar to the previous one, we obtain that 
\begin{equation*}
    \mathcal{G}\nu_{D_i}(D_t)=0, \quad i\in I, t\in T_a^{\circ}.
\end{equation*}

If $\mathcal{G}=0$, once more again this contradicts any of the cases of the Proposition (\ref{PyatovTwarock2002Theorem3.5}). On the other hand, if $\nu_{D_i}(D_t)=0$, this contradicts the fact that $\nu_{D_i}$ is an automorphism. 

\item [\rm (ii)] Suppose that $|T^{\bullet}|\geq 1$. We can write the relations relevant to $T^{\bullet}$ in the following unique form 
\begin{align}\label{relnoDS2}
    \mathcal{G} :D_iD_t: = \text{sgn}(i-t) x_iD_t, \quad i\in I, \ t\in T_b^{\bullet},
\end{align}

where $\mathcal{G}\in \{g_b^{+}, g_b^{-}, g_i+g_b^{+}, g_b^{-}-g_i \}$, depending on the case of Proposition (\ref{PyatovTwarock2002Theorem3.5}). The proof is exactly the same as in the expression (\ref{relnoDS}): we only need to change the sign of $i-t$. 
\end{enumerate}
\end{proof}

\section{Acknowledgments}

The authors gratefully acknowledge to Professors Brzezi{\'n}ski, Kr\"ahmer and Lomp for their comments and suggestions that improved the article.

The authors were supported by Faculty of Science, Universidad Nacional de Colombia - Sede Bogot\'a, Colombia [grant number 53880].

\section{Declarations}

\subsection{Conflict of interest}

The authors have no conflicts to disclose.

\end{document}